\documentclass[12pt,reqno]{amsart}
\allowdisplaybreaks[4]
\usepackage{graphicx} 
\usepackage{amssymb}
\usepackage{amsmath}
\usepackage{cite}
\usepackage{subfigure}
\usepackage{graphicx}

\newtheorem{proposition}{Proposition}

\newtheorem{theorem}{Theorem}

\begin{document}

\title[Darboux]{
The higher order Rogue Wave solutions  of the Gerdjikov-Ivanov equation}
\author{Lijuan Guo, Yongshuai Zhang, Shuwei Xu, Zhiwei wu,  Jingsong He}

\thanks{$^*$ Corresponding author: hejingsong@nbu.edu.cn,jshe@ustc.edu.cn}

\maketitle
\dedicatory { \ Department of Mathematics, Ningbo University,
Ningbo, Zhejiang 315211, P.\ R.\ China\\}

\begin{abstract}
 We construct higher order rogue wave solutions for the Gerdjikov-Ivanov equation explicitly in term of determinant expression.
 Dynamics of both soliton and non-soliton solutions is discussed. A
 family of solutions with distinct structures are presented, which
 are new to the Gerdjikov-Ivanov equation.
\end{abstract}

\section{{\bf Introduction}}

Rogue wave is one type of natural disasters first found in deep
ocean wave. Later, people observed similar phenomena in other
physical breaches, such as optical physics, plasmas, capillary waves
and so on \cite{NOW2009,Nature450,EPJST18557,PRL104104503,PRE86}.
There is a consensus that these rogue waves are the result of
modulation instability waves. Meanwhile, breather solution usually
comes from the instability of small amplitude perturbations that may
grow in size to disastrous proportions. Therefore, in mathematical
understanding, rogue wave can be treated as a limit case of Ma
soliton when the space period tends to infinity , or of the
Akhmediev breather as the time period approaches to
infinity\cite{PLA373675}. Due to both theoretical frame and reality
application, it is imperative to do further study about rogue waves,
especially, higher order rogue waves of different models.

The Peregrine soliton which is located in time-space plane is one of
the formal ways to explain the surprise phenomenon mathematically
\cite{JAMSS2516}. The solution appears from a non-zero constant and
disappears to the constant background as time approaching infinity,
but it develops a localized hump with peak amplitude three times of
average waves in the intermediate times. The object of higher order
rogue wave is now under intense discussion by different approaches
such as iterated Darboux transformation method
\cite{PRE82026602,PRE80026601}, the algebra-geometric means
\cite{EPJST185247}, and generalized Darboux transformation
\cite{PRE85026607}. Recently, the expression of the first order
rogue wave solution and the figure of the second order rouge wave
for the Gerdjikov-Ivanov (GI) equation was provided by the third and
fifth authors of present paper \cite{JMP063507}. However, the higher
order rogue wave solutions for this equation have not been studied.
The main aim of this paper is to discuss higher order rogue wave
solutions and describe their different structures.

The nonlinear Schr\"odinger equation is one of the most important
equations in physics, which can be derived from
Ablowitz-Kaup-Newell-Segur system \cite{PRL31125, SPJETP3462}.
Considering the higher order nonlinear effects, the derivative
nonlinear schr\"odignger equation with a polynomial spectral problem
of arbitrary order\cite{JPAMG143125} is regard as a model in a wide
variety of fields such as weakly  nonlinear dispersive water
waves\cite{PRS357}, nonlinear optics fibers\cite{PRA23,PRA27,JP39},
quantum field theory\cite{JPA23}, plasmas\cite{PF14}. The DNLS
equations have three generic deformations, the DNLSI equation
\cite{JMP19798}:
\begin{equation}\label{dnlsi}
{\rm i}q_{t}-q_{xx}+{\rm i}(q^2q^\ast)_{x}=0,
\end{equation}
the DNLSII equation \cite{PS20490}:
\begin{equation}\label{dnlsii}
{\rm i}q_{t}+q_{xx}+{\rm i}qq^{\ast}q_{x}=0,
\end{equation}
and the DNLSIII equation or the GI equation \cite{JPB10130}:
\begin{equation}\label{dnlsiii}
{\rm i}q_t+q_{xx}-{\rm i}q^2q^{\ast}_{x}+\frac{1}{2}q^3{q^\ast}^2=0.
\end{equation}

In many circumstances, Darboux transformation has been proved to be
one powerful methods to obtain soliton solutions \cite{MS91, Gu05},
breather solutions and rational solutions. The Darboux
transformation and its determinant expression for the GI equations
have been given in \cite{JPAMG33625, JMP063507}. But there is not a
straight extension construct higher order rogue wave solutions at
the same eigenvalues. In this paper, we take a limit technique with
respect to degenerate eigenvalues and Taylor expansion in Darboux
transformation \cite{PRE85026607,PLA166205,arxiv12093742}. Based on
this explicit method, we can further discuss the structure of
solutions, both soliton and rogue wave.

This paper is organized as following: In section 2, we review the
general $n$-fold Darboux transformation for the GI equation. In
section 3, explicit solutions are constructed, such as soltion,
breather, position solutions and higher order rogue wave with two
parameters $D_1$ and $D_2$.  By choosing different values of $D_1$
and $D_2$, we show four basic models, fundamental pattern,
triangular structure, modified triangular structure and ring
structure, and display their dynamical evolutions respectively in
section 4.  The conclusions and discussions are contained in the
final section.

\section{{\bf Darboux transformation for the Gerdjikov-Ivanov equation}}
In this section, we start with Lax pair of \eqref{dnlsiii} to construct Darboux transformation.
Considering the spectral problem
\begin{equation}\label{sys11}
\left\{
\begin{aligned}
\partial_{x}\psi&=(J\lambda^2+Q_{1}\lambda+Q_{0})\psi=U\psi,\\
\partial_{t}\psi&=(2J\lambda^4+V_{3}\lambda^3+V_{2}\lambda^2+V_{1}\lambda+V_{0})\psi=V\psi.
\end{aligned}
\right.
\end{equation}
where
\begin{equation}\label{fj1}
    \psi=\left( \begin{array}{c}
      \phi \\
      \varphi\\
     \end{array} \right)=\left( \begin{array}{c}
      \phi(x,t, \lambda) \\
      \varphi(x,t, \lambda) \\
     \end{array} \right),\nonumber\\
   \quad J= \left( \begin{array}{cc}
      -i &0 \\
      0 &i\\
   \end{array} \right),\nonumber\\
  \quad Q_{1}=\left( \begin{array}{cc}
     0 &q \\
     r &0\\
  \end{array} \right),\nonumber\\
  \quad Q_{0}=\left( \begin{array}{cc}
     -\dfrac{1}{2}i q r &0 \\
     0 &\dfrac{1}{2}i q r\\
  \end{array} \right),\nonumber\\
\end{equation}
\begin{equation}\label{fj2}
V_{3}=2Q_{1},  V_{2}=Jqr,  V_{1}=\left( \begin{array}{cc}
0 &iq_{x} \\
 -ir_{x}&0\\
\end{array} \right), V_{0}=\left( \begin{array}{cc}
\dfrac{1}{2}(rq_x-qr_x)+\dfrac{1}{4}iq^2 r^2 &0 \\
 0&-\dfrac{1}{2}(rq_x-qr_x)-\dfrac{1}{4}iq^2 r^2\\
\end{array} \right).\nonumber\\
\end{equation}
here $\lambda\in\mathbb{C}$, $\psi$ is the eigenfunction of \eqref{sys11} corresponding to the eigenvalue $\lambda$.
By the condition $U_{t}-V_{x}+[U,V]=0$ , we get
\begin{equation}\label{eq11}
\left\{
\begin{aligned}
{\rm i}q_t+q_{xx}+{\rm i}q^2r_x+\frac{1}{2}q^3r^2=0,
\\
{\rm i}r_t-r_{xx}+{\rm i}r^2q_x-\frac{1}{2}q^2r^3=0.
\end{aligned}
\right.
\end{equation}
This system admits the reduction $r=-q^{*}$, and \eqref{eq11}
becomes just one equation which is the GI equation \eqref{dnlsiii}.

From gauge transformation, we can construct new solution from initial data, i.e, if  there exist some non-singular $T$, such that
\begin{equation}\label{eq22}
\left\{
\begin{aligned}
U^{[1]}=(T_{x}+T~U)T^{-1}.
\\
V^{[1]}=(T_{t}+T~V)T^{-1}.
\end{aligned}
\right.
\end{equation}
where $U^{[1]}$ and $V^{[1]}$ have the same form as $U$ and $V$ with $q$ and $r$ replaced by certain $q^{[1]}$ and $r^{[1]}$.
Therefore, it is crucial to find an algebraic formula for $T$ instead of the \eqref{eq22}.\\

{\centering \bf 2.1 N-fold Darboux transformation for the GI system}\\

The $n-fold$ Darboux transformation for $q^{[n]}$, $r^{[n]}$ of the
GI system has been given in  in\cite{JMP063507}. Since we need this
result for our paper, we cite the main theorem as follows.
\begin{theorem}
Let $\Psi_{i}=\left(\begin{array}{c}
\phi_i\\
\varphi_i
\end{array}
\right)$ $\left(i=1,2,\cdots,n\right)$ be distinct solutions related
to $\lambda_i$ of the spectral problem , then ($q^{[n]}$,$r^{[n]}$)
given by the  following formulae are new solutions of the GI system.
\begin{equation}\label{nsolution}
q^{[n]}=q+2{\rm i}\dfrac{\Omega_{11}}{\Omega_{12}}, r^{[n]}=r-2{\rm i}\dfrac{\Omega_{21}}{\Omega_{22}}.
\end{equation}\\
Here,
(1) for $n=2k$,
\begin{equation}\label{ntt5}
\Omega_{11}=\begin{vmatrix}
$$\phi_{1}&\lambda_{1}\varphi_{1}&\ldots&\lambda_{1}^{n-3}\varphi_{1}&\lambda_{1}^{n-2}\phi_{1}&-\lambda_{1}^{n}\phi_{1}$$\\
\phi_{2}&\lambda_{2}\varphi_{2}&\ldots&\lambda_{2}^{n-3}\varphi_{2}&\lambda_{2}^{n-2}\phi_{2}&-\lambda_{2}^{n}\phi_{2}\\
\vdots&\vdots&\vdots&\vdots&\vdots&\vdots\\
\phi_{n}&\lambda_{n}\varphi_{n}&\ldots&\lambda_{n}^{n-3}\varphi_{n}&\lambda_{n}^{n-2}\phi_{n}&-\lambda_{n}^{n}\phi_{n}\\
\end{vmatrix},
\end{equation}
\begin{equation*}
\Omega_{12}=\begin{vmatrix}
\phi_{1}&\lambda_{1}\varphi_{1}&\ldots&\lambda_{1}^{n-3}\varphi_{1}&\lambda_{1}^{n-2}\phi_{1}&\lambda_{1}^{n-1}\varphi_{1}\\
\phi_{2}&\lambda_{2}\varphi_{2}&\ldots&\lambda_{2}^{n-3}\varphi_{2}&\lambda_{2}^{n-2}\phi_{2}&\lambda_{2}^{n-1}\varphi_{2}\\
\vdots&\vdots&\vdots&\vdots&\vdots&\vdots\\
\phi_{n}&\lambda_{n}\varphi_{n}&\ldots&\lambda_{n}^{n-3}\varphi_{n}&\lambda_{n}^{n-2}\phi_{n}&\lambda_{n}^{n-1}\varphi_{n}\\
\end{vmatrix},
\end{equation*}
\begin{equation*}
\Omega_{21}=\begin{vmatrix}
\varphi_{1}&\lambda_{1}\phi_{1}&\ldots&\lambda_{1}^{n-3}\phi_{1}&\lambda_{1}^{n-2}\varphi_{1}&-\lambda_{1}^{n}\varphi_{1}\\
\varphi_{2}&\lambda_{2}\phi_{2}&\ldots&\lambda_{2}^{n-3}\phi_{2}&\lambda_{2}^{n-2}\varphi_{2}&-\lambda_{2}^{n}\varphi_{2}\\
\vdots&\vdots&\vdots&\vdots&\vdots&\vdots\\
\varphi_{n}&\lambda_{n}\phi_{n}&\ldots&\lambda_{n}^{n-3}\phi_{n}&\lambda_{n}^{n-2}\varphi_{n}&-\lambda_{n}^{n}\varphi_{n}\\
\end{vmatrix},
\end{equation*}
\begin{equation*}
\Omega_{22}=\begin{vmatrix}
\varphi_{1}&\lambda_{1}\phi_{1}&\ldots&\lambda_{1}^{n-3}\phi_{1}&\lambda_{1}^{n-2}\varphi_{1}&\lambda_{1}^{n-1}\phi_{1}\\
\varphi_{2}&\lambda_{2}\phi_{2}&\ldots&\lambda_{2}^{n-3}\phi_{2}&\lambda_{2}^{n-2}\varphi_{2}&\lambda_{2}^{n-1}\phi_{2}\\
\vdots&\vdots&\vdots&\vdots&\vdots&\vdots\\
\varphi_{n}&\lambda_{n}\phi_{n}&\ldots&\lambda_{n}^{n-3}\phi_{n}&\lambda_{n}^{n-2}\varphi_{n}&\lambda_{n}^{n-1}\phi_{n}\\
\end{vmatrix};
\end{equation*}

(2) for $n=2k+1$,
\begin{equation}\label{ntt6}
\Omega_{11}=\begin{vmatrix}
\varphi_{1}&\lambda_{1}\phi_{1}&\ldots&\lambda_{1}^{n-3}\phi_{1}&\lambda_{1}^{n-2}\varphi_{1}&-\lambda_{1}^{n}\phi_{1}\\
\varphi_{2}&\lambda_{2}\phi_{2}&\ldots&\lambda_{2}^{n-3}\phi_{2}&\lambda_{2}^{n-2}\varphi_{2}&-\lambda_{2}^{n}\phi_{2}\\
\vdots&\vdots&\vdots&\vdots&\vdots&\vdots\\
\varphi_{n}&\lambda_{n}\phi_{n}&\ldots&\lambda_{n}^{n-3}\phi_{n}&\lambda_{n}^{n-2}\varphi_{n}&-\lambda_{n}^{n}\phi_{n}\\
\end{vmatrix},
\end{equation}
\begin{equation*}
\Omega_{12}=\begin{vmatrix}
\varphi_{1}&\lambda_{1}\phi_{1}&\ldots&\lambda_{1}^{n-3}\phi_{1}&\lambda_{1}^{n-2}\varphi_{1}&\lambda_{1}^{n-1}\varphi_{1}\\
\varphi_{2}&\lambda_{2}\phi_{2}&\ldots&\lambda_{2}^{n-3}\phi_{2}&\lambda_{2}^{n-2}\varphi_{2}&\lambda_{2}^{n-1}\varphi_{2}\\
\vdots&\vdots&\vdots&\vdots&\vdots&\vdots\\
\varphi_{n}&\lambda_{n}\phi_{n}&\ldots&\lambda_{n}^{n-3}\phi_{n}&\lambda_{n}^{n-2}\varphi_{n}&\lambda_{n}^{n-1}\varphi_{n}\\
\end{vmatrix},
\end{equation*}

\begin{equation*}
\Omega_{21}=\begin{vmatrix}
\phi_{1}&\lambda_{1}\varphi_{1}&\ldots&\lambda_{1}^{n-3}\varphi_{1}&\lambda_{1}^{n-2}\phi_{1}&-\lambda_{1}^{n}\varphi_{1}\\
\phi_{2}&\lambda_{2}\varphi_{2}&\ldots&\lambda_{2}^{n-3}\varphi_{2}&\lambda_{2}^{n-2}\phi_{2}&-\lambda_{2}^{n}\varphi_{2}\\
\vdots&\vdots&\vdots&\vdots&\vdots&\vdots\\
\phi_{n}&\lambda_{n}\varphi_{n}&\ldots&\lambda_{n}^{n-3}\varphi_{n}&\lambda_{n}^{n-2}\phi_{n}&-\lambda_{n}^{n}\varphi_{n}\\
\end{vmatrix},
\end{equation*}
\begin{equation*}
\Omega_{22}=\begin{vmatrix}
\phi_{1}&\lambda_{1}\varphi_{1}&\ldots&\lambda_{1}^{n-3}\varphi_{1}&\lambda_{1}^{n-2}\phi_{1}&\lambda_{1}^{n-1}\phi_{1}\\
\phi_{2}&\lambda_{2}\varphi_{2}&\ldots&\lambda_{2}^{n-3}\varphi_{2}&\lambda_{2}^{n-2}\phi_{2}&\lambda_{2}^{n-1}\phi_{2}\\
\vdots&\vdots&\vdots&\vdots&\vdots&\vdots\\
\phi_{n}&\lambda_{n}\varphi_{n}&\ldots&\lambda_{n}^{n-3}\varphi_{n}&\lambda_{n}^{n-2}\phi_{n}&\lambda_{n}^{n-1}\phi_{n}\\
\end{vmatrix}.
\end{equation*}
\end{theorem}

To solve the GI equation, we need to preserve the reduction
condition $r=-q^*$, i.e, under n steps of Darboux transformation,
$r^{[n]}=-{q^{[n]}}^*$. Therefore, we can choose eigenfunctions
$\psi_k=\left(
\begin{array}{c}
\phi_k\\
\varphi_k
\end{array} \right)$ as follows \cite{JMP063507}:
\begin{enumerate}
\item $\lambda_k=-\lambda_k^*$, and $\psi_k=\left( \begin{array}{c}
\phi_k\\
\varphi_k
\end{array} \right)=\left(
\begin{array}{c}
\varphi_k^*\\
\phi_k^*
\end{array}
\right)$.

\item $\lambda_{2k}=-\lambda_{2k-1}^*$, and $\psi_{2k}=\left( \begin{array}{c}
\phi_{2k}\\
\varphi_{2k}
\end{array} \right)=
\left(\begin{array}{c}
-\varphi_{2k-1}^*\\
\phi_{2k-1}^*
\end{array}\right).$
\end{enumerate}

By this choice, we can construct soliton, breather, position
solutions of the GI equation. But to get higher order rogue wave, we
need to modify above Darboux transformation. At that time, the
eigenfunctions and eigenvalues are no longer arbitrary, and only
condition (2) works. We will discuss them in detail in the next
section.

\section{{\bf Solutions of the Gerdjikov-Ivanov equation}}
{\bf 3.1 Solution with vanishing boundary condition}\\

We start with the soliton solutions by \eqref{nsolution}. Let $q=0$,
solutions of the spectral problem \eqref{sys11} with eigenvalues
$\lambda_k$ are solved as
\begin{equation}\label{fun1}
\psi_k=\left(\begin{array}{c}\phi_k\\\varphi_k\end{array}\right),
\quad \phi_k={\rm exp}(-{\rm i}(\lambda_k^{2}x+2\lambda_k^{4}t)),\quad
\varphi_k={\rm exp}({\rm i}(\lambda_k^{2}x+2\lambda_k^{4}t)).
\end{equation}
Example:

\begin{itemize}
  \item  $n=1$.
   $\lambda_1^{*}=-\lambda_1$,$\lambda_1={\rm i}\beta_1$ and  $\psi_1=\left(\begin{array}{c}\phi_1\\\phi_1^{*}\end{array}\right)$, then
  $$q^{[1]}= 2\,\beta_1\,{{\rm exp}^{-2\,{\rm i}{\beta_1}^{2} \left( -x+2\,{\beta_1}^{2}t
  \right). }}
  $$
 This is a plane wave with constant amplitude.
  \item $n=2$.

  \begin{enumerate}

  \item Let $\lambda_2=\lambda_1^{*}$, $\lambda_1=\alpha_1+{\rm i}\beta_1$ and
  $\psi_1=\left(\begin{array}{c}
  -\phi_1\\
  \varphi_1
  \end{array}\right)
  $
  $\psi_2=\left(\begin{array}{c}
  -\varphi_1^{*} \\
  \phi_1^*
  \end{array}\right)
  $, then
  \begin{equation}\label{q[2]}
  q^{[2]}=-8\,{\frac {{{\rm e}^{F_{{1}}}}\alpha_{{1}}\beta_{{1}}}{-{{\rm e}^{F_{
  {2}}}}\alpha_{{1}}+{\rm i}{{\rm e}^{F_{{2}}}}\beta_{{1}}-{{\rm e}^{-F_{{2}}}
  }\alpha_{{1}}-{\rm i}{{\rm e}^{-F_{{2}}}}\beta_{{1}}}}
   \end{equation}
  with
  \begin{equation}
  \begin{split}
   F_1=&-2\,{\rm i} \left( {\alpha_{{1}}}^{2}x+2\,{\alpha_{{1}}}^{4}t-12\,{\alpha_{{
      1}}}^{2}{\beta_{{1}}}^{2}t-{\beta_{{1}}}^{2}x+2\,{\beta_{{1}}}^{4}t
  \right),\\
  F_2=&4\,\alpha_{{1}}\beta_{{1}} \left( 4\,t{\alpha_{{1}}}^{2}-4\,t{\beta_{{
  1}}}^{2}+x \right).
  \end{split}
  \end{equation}
   It is a line soliton, and its trajectory is
  $$
  x=-4\,t \left( \alpha_{{1}}-\beta_{{1}} \right)  \left( \alpha_{{1}}+
  \beta_{{1}} \right),
  $$
  on the (x, t) plane. Let $\alpha_1 \rightarrow0$ in (\ref{q[2]}), we obtain a rational travelling solution
  \begin{equation}\label{q[21]}
  q^{[2]}=4\,{\frac {\beta_{{1}}{{\rm e}^{-2\,{\rm i}{\beta_{{1}}}^{2} \left( -x+2\,t{
  \beta_{{1}}}^{2} \right) }}}{1-4\,{\rm i}{\beta_{{1}}}^{2}x+16\,{\rm i}{\beta_{{1}
  }}^{4}t}},
  \end{equation}
  with an arbitrary real constant $\beta_1$. Its trajectory is
  defined explicitly by
  $$
  x=4\beta_1^{2}t
  $$
  on the (x, t) plane.  The above two solutions is plotted in Fig. \ref{fig.soliton}(a) and Fig. \ref{fig.soliton}(b).

  \item Let  $\lambda_1=-\lambda_1^*={\rm i}\beta_1$, $\lambda_2=-\lambda_2^*={\rm i}\beta_2$,
  then
  $\psi_1=\left(\begin{array}{c}
  \phi_1\\
  \phi_1^*
  \end{array}
  \right)$,
  $\psi_2=\left(\begin{array}{c}
  \phi_2\\
  \phi_2^*
  \end{array}
  \right)$.
 we obtain a
soliton solution
\begin{equation}\label{q[22]}
q^{[2]}=-2\,{\frac {{{\rm e}^{-{\rm i} \left( -{\beta_{{1}}}^{2}x+2\,{\beta_{{1}}}^{
4}t-{\beta_{{2}}}^{2}x+2\,{\beta_{{2}}}^{4}t \right) }} \left( -{\beta
_{{2}}}^{2}+{\beta_{{1}}}^{2} \right) }{{{\rm e}^{-{\rm i} \left( \beta_{{1}
}-\beta_{{2}} \right)  \left( \beta_{{1}}+\beta_{{2}} \right)  \left(
2\,t{\beta_{{1}}}^{2}-x+2\,{\beta_{{2}}}^{2}t \right) }}\beta_{{2}}-
\beta_{{1}}{{\rm e}^{{\rm i} \left( \beta_{{1}}-\beta_{{2}} \right)  \left(
\beta_{{1}}+\beta_{{2}} \right)  \left( 2\,t{\beta_{{1}}}^{2}-x+2\,{
\beta_{{2}}}^{2}t \right) }}}}.
\end{equation}
Its trajectory is defined explicitly as
$$
x=2\,t \left( {\beta_{{1}}}^{2}+{\beta_{{2}}}^{2} \right)
$$
in the (x, t) plane.
It is plotted in Fig. \ref{fig.soliton}(c).
\end{enumerate}

\item $n=3$. Let $\lambda_i=-\lambda_i^*$, $i=1,\mbox{ }2,\mbox{ }3$, we get a one-soliton solution. While  a $\lambda_1=-\lambda_1^*$, $\lambda_3=\lambda_2^*$,
we obtain a quasi-periodic solution. Since the analytic expressions are clumsy and tedious,  we plot a particular quasi-periodic solution  in Fig. 2(a).

\item $n=4$. Under the different reduction conditions, we obtain three kinds of solutions.
One of them is a soliton solution under the quasi-periodic background, the others are second soliton solutions.
 A soliton solution under the quasi-periodic background is plotted in Fig. 2(b).
\end{itemize}

{\bf Remark}. If the eigenvalues share the same value, $T$ becomes degenerate. Next, we will consider the degenerate cases by limit technique.

\bigskip

Given $\psi=\psi(\lambda)$ , define $\psi[i,j,k]$ by
\begin{equation}
\lambda^j\psi=\psi[i,j,0]+\psi[i,j,1]\epsilon+\psi[i,j,2]\epsilon^2+\cdots+\psi[i,j,k]\epsilon^k+\cdots,
\end{equation}
with
$$
\psi[i,j,k]=\frac{1}{k!}\frac{\partial^k}{\partial\epsilon^k}((\lambda_i+\epsilon)^j\psi(\lambda_i+\epsilon)).
$$
In other words, $\psi[i,j,k]$ is the coefficient of $\epsilon^k$ if we expand $\lambda^j\psi$ at $\lambda=\lambda_i$
by Taylor expansion.

Now, we consider the degenerate case of $n=4$. According to the formula,
\begin{equation}\label{q[4]}
q^{[4]}=\frac{2{\rm i}\delta_{11}}{\delta_{12}}
\end{equation}
with
$$\delta_{11}=
\left| \begin{array}{cccc} \phi_{{1}}&\lambda_{{1}}\psi_{{1}}&{\lambda_{{1}}}^{2}\phi_{{1}}&-{\lambda_{{1}}}^{4}\phi_{{1}}
\\ \noalign{\medskip}\phi_{{2}}&\lambda_{{2}}\psi_{{2}}&{\lambda_{{2}}
}^{2}\phi_{{2}}&-{\lambda_{{2}}}^{4}\phi_{{2}}\\ \noalign{\medskip}
\phi_{{3}}&\lambda_{{3}}\psi_{{3}}&{\lambda_{{3}}}^{2}\phi_{{3}}&-{
\lambda_{{3}}}^{4}\phi_{{3}}\\ \noalign{\medskip}\phi_{{4}}&\lambda_{{
4}}\psi_{{4}}&{\lambda_{{4}}}^{2}\phi_{{4}}&-{\lambda_{{4}}}^{4}\phi_{
{4}}\end{array}  \right|,
$$
$$\delta_{12}=
\left| \begin{array}{cccc} \phi_{{1}}&\lambda_{{1}}\psi_{{1}}&{\lambda_{{1}}}^{2}\phi_{{1}}&{\lambda_{{1}}}^{3}\psi_{{1}}
\\ \noalign{\medskip}\phi_{{2}}&\lambda_{{2}}\psi_{{2}}&{\lambda_{{2}}
}^{2}\phi_{{2}}&{\lambda_{{2}}}^{3}\psi_{{2}}\\ \noalign{\medskip}\phi
_{{3}}&\lambda_{{3}}\psi_{{3}}&{\lambda_{{3}}}^{2}\phi_{{3}}&{\lambda_
{{3}}}^{3}\psi_{{3}}\\ \noalign{\medskip}\phi_{{4}}&\lambda_{{4}}\psi_
{{4}}&{\lambda_{{4}}}^{2}\phi_{{4}}&{\lambda_{{4}}}^{3}\psi_{{4}}
\end{array} \right|.
$$

Notice that these eigenfunction vectors $\psi_i(i=1,3)$ have the same formula except only different eigenvalue $\lambda_i$.
Thus, if we set $\lambda_1=\lambda_3$, the determinants are linearly dependent. Then, we will obtain a trivial solution $q^{[4]}=0$. To get non-trivial solution, we need to adopt the Taylor expansion technique.

The main process:
{\bf i)} For the first(second) row, we can substitute the eigenvalue $\lambda_1$($\lambda_2=\lambda_1^*$) and eigenfunction $\psi_1$($\psi_2$) directly.
{\bf ii)} $\lambda_3\rightarrow\lambda_1$($\lambda_4\rightarrow\lambda_2$), we expand the elements of the third(forth) row at $\lambda_3=\lambda_1+\epsilon$($\lambda_4=\lambda_2+\epsilon$)
by Taylor expansion at first, then subtract the first(second) row from the third(forth) row.
{\bf iii)} Taking $\epsilon\rightarrow0$,
the terms with higher order of $\epsilon$ will vanish.
and the new
expression of $q^{[4]}$ is obtained as the same form as eq. (\ref{q[4]}), but with the values for $\delta_{11}$ and
$\delta_{12}$ given by
\begin{equation}\nonumber
\delta_{11}=\left|
\begin{array}{cccc}
\phi[1,0,0]&\varphi[1,1,0]&\phi[1,2,0]&-\phi[1,4,0]\\
\phi[2,0,0]&\varphi[2,1,0]&\phi[2,2,0]&-\phi[2,4,0]\\
\phi[1,0,1]&\varphi[1,1,1]&\phi[1,2,1]&-\phi[1,4,1]\\
\phi[2,0,1]&\varphi[2,1,1]&\phi[2,2,1]&-\phi[2,4,1]
\end{array}\right|,
\end{equation}
\begin{equation}
\delta_{12}=\left|
\begin{array}{cccc}
\phi[1,0,0]&\varphi[1,1,0]&\phi[1,2,0]&\varphi[1,3,0]\\
\phi[2,0,0]&\varphi[2,1,0]&\phi[2,2,0]&\varphi[2,3,0]\\
\phi[1,0,1]&\varphi[1,1,1]&\phi[1,2,1]&\varphi[1,3,1]\\
\phi[2,0,1]&\varphi[2,1,1]&\phi[2,2,1]&\varphi[2,3,1]
\end{array}
\right|.
\end{equation}

Substituting $\lambda_1=\alpha_1+{\rm i}\beta_1$ and eigenfunctions of eq. (\ref{fun1}) into the above formula. we obtain a positon solution
\begin{equation}\label{2p}
q^{[4]}=\frac{G_1+{\rm i}G_2}{G_3+{\rm i}G_4},
\end{equation}
where
\begin{equation}
\begin{split}
G_1=&\alpha_{{1}}\beta_{{1}},\\
G_2=&-16\cos(F_2)\alpha_1^3(-32\beta_1^3\alpha_1^{2}t\sinh(F_1)+16\beta_1\alpha_1^4t\sinh(F_1)
     +4\beta_1^3x\sinh(F_1)+4\alpha_1^2\beta_1x\sinh(F_1)\\
     &-48\beta_1^5t\sinh(F_1)
     -\alpha_1\cosh(F_1))\beta_1,\\
G_3=&16\,{\alpha_{{1}}}^{4}{\beta_{{1}}}^{2}x+64\,{\alpha_{{1}}}^{6}t{\beta
     _{{1}}}^{2}+2\,\alpha_{{1}}{\beta_{{1}}}^{3}\sinh \left( 2\,F_{{2}}
     \right) -384\,{\alpha_{{1}}}^{4}{\beta_{{1}}}^{4}t-16\,{\alpha_{{1}}}
     ^{2}{\beta_{{1}}}^{4}x\\
     &+2\,{\alpha_{{1}}}^{3}\beta_{{1}}\sinh \left( 2
     \,F_{{2}} \right) +64\,{\alpha_{{1}}}^{2}{\beta_{{1}}}^{6}t,\\
G_4=&-256\,{\alpha_{{1}}}^{2}x{\beta_{{1}}}^{8}t+32\,{\alpha_{{1}}}^{2}{x}^
    {2}{\beta_{{1}}}^{6}+{\alpha_{{1}}}^{4}+64\,{\beta_{{1}}}^{4}{x}^{2}{
    \alpha_{{1}}}^{4}-256\,{\alpha_{{1}}}^{4}x{\beta_{{1}}}^{6}t+256\,{
    \alpha_{{1}}}^{8}{\beta_{{1}}}^{2}xt\\
    &+3072\,{\beta_{{1}}}^{6}{t}^{2}{
    \alpha_{{1}}}^{6}+2048\,{\beta_{{1}}}^{4}{t}^{2}{\alpha_{{1}}}^{8}+32
    \,{\alpha_{{1}}}^{6}{x}^{2}{\beta_{{1}}}^{2}+512\,{\alpha_{{1}}}^{10}{
    t}^{2}{\beta_{{1}}}^{2}+512\,{\beta_{{1}}}^{10}{t}^{2}{\alpha_{{1}}}^{
    2}\\
    &+2048\,{\alpha_{{1}}}^{4}{t}^{2}{\beta_{{1}}}^{8}+{\beta_{{1}}}^{4}+
    256\,{\beta_{{1}}}^{4}x{\alpha_{{1}}}^{6}t+{\alpha_{{1}}}^{4}\cosh
    \left( 2\,F_{{2}} \right) -{\beta_{{1}}}^{4}\cosh \left( 2\,F_{{2}}
    \right),\\
F_1=&4\,\alpha_{{1}}\beta_{{1}} \left( x+4\,t{\alpha_{{1}}}^{2}-4\,t{\beta_
{{1}}}^{2} \right),\\
F_2=&-24\,{\alpha_{{1}}}^{2}t{\beta_{{1}}}^{2}+4\,{\beta_{{1}}}^{4}t+4\,{
\alpha_{{1}}}^{4}t-2\,{\beta_{{1}}}^{2}x+2\,{\alpha_{{1}}}^{2}x.
\end{split}\nonumber
\end{equation}
Its dynamical evolution is plotted in Fig. 3(a).

If we let $\alpha_1\rightarrow0$ in above procedure, we will get the second order rational solution.
With these parameters, the general solution can be given as following:
\begin{equation}\label{2r}
q^{[4]}=\frac{G_1}{G_2+{\rm i}G_3}{\rm exp}({-2\,{\rm i}{\beta_{{1}}}^{2} \left( -x+2\,t{\beta_{{1}}}^{2}
 \right) }),
\end{equation}
with
\begin{equation}\label{2rational}
\begin{split}\nonumber
G_1=&-8\beta_1(-12{\rm i}\beta_1^2x+768{\rm i}\beta_1^8x^2t+4096{\rm i}\beta_1^{12}t^3
     -64{\rm i}\beta_1^6x^3-48{\rm i}\beta_1^4t-3072{\rm i}\beta_1^{10}xt^2-3
     +2304t^2\beta_1^8\\
     &-768x\beta_1^6t+48\beta_1^4x^2),\\
G_2=&3-768\,x{\beta_{{1}}}^{6}t+4096\,{\beta_{{1}}}^{10}t{x}^{3}-24576\,{
\beta_{{1}}}^{12}{x}^{2}{t}^{2}+65536\,{\beta_{{1}}}^{14}x{t}^{3}+4608
\,{t}^{2}{\beta_{{1}}}^{8}-65536\,{\beta_{{1}}}^{16}{t}^{4}\\
&-256\,{
\beta_{{1}}}^{8}{x}^{4}+96\,{\beta_{{1}}}^{4}{x}^{2},\\
G_3=&-12288\,{\beta_{{1}}}^{10}x{t}^{2}+16384\,{\beta_{{1}}}^{12}{t}^{3}-
256\,{\beta_{{1}}}^{6}{x}^{3}-48\,{\beta_{{1}}}^{2}x+3072\,{\beta_{{1}
}}^{8}{x}^{2}t+576\,{\beta_{{1}}}^{4}t.
\end{split}
\end{equation}

The equations (\ref{2p}) and (\ref{2r}) represent the interaction of solitons and rational solitons,
respectively. A simplify analysis shows that they possess phase shift when $t\rightarrow\pm\infty$. Which is
different from  general $2$-soliton solutions. They are shown in Fig. 3.


This process works for general $n$.
\begin{theorem}
For $n=2k$, $\psi=\left(\begin{array}
  {c}\phi\\
  \varphi
\end{array}\right)$ is eigenfunction vector of the GI system,
the expression of $n$-positon solutions in terms of determinant is obtained:
\begin{equation}\label{nposoliton}
q^{[n]}=2{\rm i}\frac{\delta_{11}}{\delta_{12}},
\end{equation}
where
\begin{equation}
\delta_{11}=\left|
\begin{array}{cccccc}
\phi[1,0,0]&\varphi[1,1,0]&\cdots&\varphi[1,n-3,0]&\phi[1,n-2,0]&-\phi[1,n,0]\\
\phi[2,0,0]&\varphi[2,1,0]&\cdots&\varphi[2,n-3,0]&\phi[2,n-2,0]&-\phi[2,n,0]\\
\phi[1,0,1]&\varphi[1,1,1]&\cdots&\varphi[1,n-3,1]&\phi[1,n-2,1]&-\phi[1,n,1]\\
\phi[2,0,1]&\varphi[2,1,1]&\cdots&\varphi[2,n-3,1]&\phi[2,n-2,1]&-\phi[2,n,1]\\
\vdots&\vdots&\vdots&\vdots&\vdots&\vdots\\
\phi[1,0,k-1]&\varphi[1,1,k-1]&\cdots&\varphi[1,n-3,k-1]&\phi[1,n-2,k-1]&-\phi[1,n,k-1]\\
\phi[2,0,k-1]&\varphi[2,1,k-1]&\cdots&\varphi[2,n-3,k-1]&\phi[2,n-2,k-1]&-\phi[2,n,k-1]\nonumber
\end{array}
\right|,
\end{equation}
\begin{equation}
\delta_{12}=\left|
\begin{array}{cccccc}
  \phi[1,0,0]&\varphi[1,1,0]&\cdots&\varphi[1,n-3,0]&\phi[1,n-2,0]&\varphi[1,n-1,0]\\
  \phi[2,0,0]&\varphi[2,1,0]&\cdots&\varphi[2,n-3,0]&\phi[2,n-2,0]&\varphi[2,n-1,0]\\
  \phi[1,0,1]&\varphi[1,1,1]&\cdots&\varphi[1,n-3,1]&\phi[1,n-2,1]&\varphi[1,n-1,1]\\
  \phi[2,0,1]&\varphi[2,1,1]&\cdots&\varphi[2,n-3,1]&\phi[2,n-2,1]&\varphi[2,n-1,1]\\
  \vdots&\vdots&\vdots&\vdots&\vdots&\vdots\\
  \phi[1,0,k-1]&\varphi[1,1,k-1]&\cdots&\varphi[1,n-3,k-1]&\phi[1,n-2,k-1]&\varphi[1,n-1,k-1]\\
  \phi[2,0,k-1]&\varphi[2,1,k-1]&\cdots&\varphi[2,n-3,k-1]&\phi[2,n-2,k-1]&\varphi[2,n-1,k-1]\nonumber
\end{array}
\right|.
\end{equation}
\end{theorem}


\begin{proof} For the entries in the first column of $\Omega_{11} \eqref{ntt5}$,
\begin{equation*}
\begin{aligned}
  \phi_1&=\phi[1,0,0],\\
  \phi_2&=\phi[2,0,0],\\
  \phi_3&=\phi[1,0,0]+\phi[1,0,1]\epsilon,\\
  \phi_4&=\phi[2,0,0]+\phi[2,0,1]\epsilon,\\
  \vdots\\
  \phi_{n-1}&=\phi[1,0,0]+\phi[1,0,1]\epsilon+\phi[1,0,2]\epsilon^2+\cdots+\phi[1,0,k-1]\epsilon^{k-1},\\
  \phi_{n}&=\phi[2,0,0]+\phi[2,0,1]\epsilon+\phi[2,0,2]\epsilon^2+\cdots+\phi[2,0,k-1]\epsilon^{k-1},
\end{aligned}
\end{equation*}
Taking the similar procedure to the other entries in $\Omega_{11}, \Omega_{12}$. Finally, the $q^{[n]}$ can be
obtained through simple calculation and certain limit.
\end{proof}


Let $\alpha_1\rightarrow0$ in above formula, we will get the $n$-rational solution.\\

{\bf 3.2 Solutions with non-vanishing boundary condition}\\

In this section, we consider solutions from non-trivial seed, and this will give higher order rogue wave. In general, we start with  $ q = ce^{{\rm i}[ax-(a^{2}+c^{2}a-\frac{c^{4}}{2})t]}$, $a,c\in\mathbb{C}^2$. Then the corresponding eigenfunctions $\psi_k$ associated with $\lambda_k$\\

\begin{eqnarray}\label{efnz}
\left(\mbox{\hspace{-0.2cm}} \begin{array}{c}
 \phi_{k}(\lambda_{k})\\
 \varphi_{k}(\lambda_{k})\\
\end{array}\mbox{\hspace{-0.2cm}}\right)\mbox{\hspace{-0.2cm}}=\mbox{\hspace{-0.2cm}}\left(\mbox{\hspace{-0.2cm}}\begin{array}{c}
 D_1\varpi_1(\lambda_{k})[1,k]+D_2\varpi_2(\lambda_{k})[1,k]-D_2\varpi_1^{\ast}({\lambda_{k}^{\ast}})[2,k]-D_1\varpi_2^{\ast}({\lambda_{k}^{\ast}})[2,k]\\
 D_1\varpi_1(\lambda_{k})[2,k]+D_2\varpi_2(\lambda_{k})[2,k]+D_2\varpi_1^{\ast}({\lambda_{k}^{\ast}})[1,k]+D_1\varpi_2^{\ast}({\lambda_{k}^{\ast}})[1,k]\\
\end{array}\mbox{\hspace{-0.2cm}}\right),
\end{eqnarray}
where

\begin{eqnarray*}
\left(\mbox{\hspace{-0.2cm}}\begin{array}{c}
 \varpi_1(\lambda_{k})[1,k]\\
 \varpi_1(\lambda_{k})[2,k]\\
\end{array}\mbox{\hspace{-0.2cm}}\right)\mbox{\hspace{-0.2cm}}=\mbox{\hspace{-0.2cm}}\left(\mbox{\hspace{-0.2cm}}\begin{array}{c}
\exp(-c_1{(ta-x-2t{\lambda_{k}}^2)}+\frac{{\rm i}(ax-(a^{2}+c^{2}a-\frac{c^{4}}{2})t)}{2}) \\
\frac{({\rm i}a-{\rm i}c^{2}+2{\rm i}{\lambda_{k}}^2+2c_1)}{2{\lambda_{k}}c}\exp({-c_1(ta-x-2t{\lambda_{k}}^2)}-\frac{{\rm i}(ax-(a^{2}+c^{2}a-\frac{c^{4}}{2})t)}{2}) \\\\
\end{array}\mbox{\hspace{-0.2cm}}\right),\\
\end{eqnarray*}
\begin{eqnarray*}
\left(\mbox{\hspace{-0.2cm}}\begin{array}{c}
 \varpi_2(\lambda_{k})[1,k]\\
 \varpi_2(\lambda_{k})[2,k]\\
\end{array}\mbox{\hspace{-0.2cm}} \right)\mbox{\hspace{-0.2cm}}=\mbox{\hspace{-0.3cm}}
\left(\mbox{\hspace{-0.3cm}}\begin{array}{c}
\exp({c_1(ta-x-2t{\lambda_{k}}^2)}+\frac{{\rm i}(ax-(a^{2}+c^{2}a-\frac{c^{4}}{2})t)}{2}) \\
\frac{({\rm i}a-{\rm i}c^{2}+2{\rm i}{\lambda_{k}}^2-2c_1)}{2{\lambda_{k}}c}\exp({c_1(ta-x-2t{\lambda_{k}}^2)}-\frac{{\rm i}(ax-(a^{2}+c^{2}a-\frac{c^{4}}{2})t)}{2}) \\
\end{array}\mbox{\hspace{-0.3cm}}\right)\mbox{\hspace{-0.1cm}},\\
\end{eqnarray*}
\begin{eqnarray*}
\varpi_1(\lambda_{k})=
\left( \begin{array}{c}
 \varpi_1(\lambda_{k})[1,k]\\
 \varpi_1(\lambda_{k})[2,k]\\
\end{array} \right),~~~~~
\varpi_2(\lambda_{k})=
\left( \begin{array}{c}
 \varpi_2(\lambda_{k})[1,k]\\
 \varpi_2(\lambda_{k})[2,k]\\
\end{array} \right),
\end{eqnarray*}
\begin{equation}\label{c1}
c_1=\frac{\sqrt{-c^{4}+2c^{2}a-a^{2}-4a{\lambda_{k}}^{2}-4{\lambda_{k}}^{4}}}{2}.
\end{equation}

 In the following, we set $D_1=1$ and $D_2=1$. Under the circumstance, Xu {\sl et. al.}\cite{JMP063507} have obtained two kinds of breather solutions. One is time periodic breather, and the other is space periodic breather.  Besides, the first order and the second order rogue wave solutions have also been obtained by a certain limit from breather solution. However, the method
 is difficult to calculate the higher order rogue waves. What we do here is to take the limit in the determinant expression of solution for the GI equation directly, which is the same method as we do for the NLS equation\cite{arxiv12093742}.


\begin{proposition}
  Let $n=2k$, $\psi=\left(\begin{array}
  {c}\phi\\
  \varphi
 \end{array}\right)$ is eigenfunction vector,
assuming
$$\lambda_1=\frac{1}{2}\sqrt{c^2-2a}-\frac{1}{2}{\rm i}c$$
and
$$\lambda_2=\frac{1}{2}\sqrt{c^2-2a}+\frac{1}{2}{\rm i}c,$$
the the formula of the $n$-th order rogue solution is obtained
\begin{equation}\label{nRW}
  q^{[n]}=q+2{\rm i}\frac{\delta_{11}}{\delta_{12}},
\end{equation}
where
\begin{equation}
\delta_{11}=\left|
\begin{array}{cccccc}
\phi[1,0,1]&\varphi[1,1,1]&\cdots&\varphi[1,n-3,1]&\phi[1,n-2,1]&-\phi[1,n,1]\\
\phi[2,0,1]&\varphi[2,1,1]&\cdots&\varphi[2,n-3,1]&\phi[2,n-2,1]&-\phi[2,n,1]\\
\phi[1,0,2]&\varphi[1,1,2]&\cdots&\varphi[1,n-3,2]&\phi[1,n-2,2]&-\phi[1,n,2]\\
\phi[2,0,2]&\varphi[2,1,2]&\cdots&\varphi[2,n-3,2]&\phi[2,n-2,2]&-\phi[2,n,2]\\
\vdots&\vdots&\vdots&\vdots&\vdots&\vdots\\
\phi[1,0,k]&\varphi[1,1,k]&\cdots&\varphi[1,n-3,k]&\phi[1,n-2,k]&-\phi[1,n,k]\\
\phi[2,0,k]&\varphi[2,1,k]&\cdots&\varphi[2,n-3,k]&\phi[2,n-2,k]&-\phi[2,n,k]\nonumber
\end{array}
\right|,
\end{equation}
\begin{equation*}
  \delta_{12}=\left|\begin{array}{cccccc}
  \phi[1,0,1]&\varphi[1,1,1]&\cdots&\varphi[1,n-3,1]&\phi[1,n-2,1]&\varphi[1,n-1,1]\\
  \phi[2,0,1]&\varphi[2,1,1]&\cdots&\varphi[2,n-3,1]&\phi[2,n-2,1]&\varphi[2,n-1,1]\\
  \phi[1,0,2]&\varphi[1,1,2]&\cdots&\varphi[1,n-3,2]&\phi[1,n-2,2]&\varphi[1,n-1,2]\\
  \phi[2,0,2]&\varphi[2,1,2]&\cdots&\varphi[2,n-3,2]&\phi[2,n-2,2]&\varphi[2,n-1,2]\\
  \vdots&\vdots&\vdots&\vdots&\vdots&\vdots\\
  \phi[1,0,k]&\varphi[1,1,k]&\cdots&\varphi[1,n-3,k]&\phi[1,n-2,k]&\varphi[1,n-1,k]\\
  \phi[2,0,k]&\varphi[2,1,k]&\cdots&\varphi[2,n-3,k]&\phi[2,n-2,k]&\varphi[2,n-1,k]
  \end{array}\right|.
\end{equation*}
\end{proposition}


\begin{proof} For the entries in the first column of $\Omega_{11} \eqref{ntt5}$,
\begin{equation*}
\begin{aligned}
  \phi_1&=\phi[1,0,1]\epsilon,\\
  \phi_2&=\phi[2,0,1]\epsilon,\\
  \phi_3&=\phi[1,0,1]\epsilon+\phi[1,0,2]\epsilon^2,\\
  \phi_4&=\phi[2,0,1]\epsilon+\phi[2,0,2]\epsilon^2,\\
  \vdots\\
  \phi_{n-1}&=\phi[1,0,1]\epsilon+\phi[1,0,2]\epsilon^2+\cdots+\phi[1,0,k]\epsilon^k,\\
  \phi_{n}&=\phi[2,0,1]\epsilon+\phi[2,0,2]\epsilon^2+\cdots+\phi[2,0,k]\epsilon^k,
\end{aligned}
\end{equation*}
Taking the similar procedure to the other entries in $\Omega_{11}, \Omega_{12}$. Finally, the $q^{[n]}$ can be
obtained through simple calculation and certain limit.
\end{proof}

As application, we give some explicit  solutions next.
\begin{itemize}
\item  first order rogue wave\\
For $n=2$, we get the first order rogue wave.
\begin{equation}\label{1RW}
q^{[2]}=\frac{G_1}{G_2}{\rm exp}\left(\frac{1}{2}\,{\rm i} \left(2\,ax-2\,t{a}^{2}-2\,t{c}^{2}a+t{c}^{4} \right)\right) ,
\end{equation}
with
\begin{equation}
\begin{split}
  G_1=&-8\,{c}^{2}{a}^{3}{t}^{2}+12\,{a}^{2}{c}^{4}{t}^{2}+8\,{c}^{2}{a}^{2}t
       x-8\,{c}^{4}tax-2\,{c}^{2}a{x}^{2}+12{\rm i}a{c}^{2}t-6\,a{c}^{6}{t}^{2}-3
       \\&-2{\rm i}{c}^{2}x+2\,{c}^{4}{x}^{2}
       -6{\rm i}{c}^{4}t+2\,{c}^{8}{t}^{2},\\
  G_2=&-8\,{c}^{2}{a}^{3}{t}^{2}+12\,{a}^{2}{c}^{4}{t}^{2}+8\,{c}^{2}{a}^{2}t
       x-8\,{c}^{4}tax-2\,{c}^{2}a{x}^{2}+4{\rm i}{c}^{2}ta-6\,a{c}^{6}{t}^{2}+1\\
       &-2{\rm i}{c}^{2}x+2\,{c}^{4}{x}^{2}
       +2{\rm i}{c}^{4}t+2\,{c}^{8}{t}^{2}.
  \end{split}\nonumber
\end{equation}

Use this method, we also get similar result of \cite{JMP063507}. Moreover, our method is general to product higher order rogue wave.
Assuming $x\rightarrow\infty$, $t\rightarrow\infty$, then $q^{[2]}\rightarrow{c^2}$. The maximum amplitude of
$q^{[2]}$ occurs at origin and is equal to $9c^2$. A first order rogue wave with particular parameters is shown in Fig. 4.\\


\item Higher order Rogue Wave

Generally, the expression of rogue wave becomes more complicated with increasing $n$. For convenience, we use numerical
simulations to discuss the higher order rogue wave. In the following, set $a=0$ and $c=1$.

When $n=4$, we can obtain the second order rogue wave according to the formula (\ref{nRW}).
\begin{equation}
q^{[4]}=\frac{G_1}{G_2}{\rm exp}\left(\frac{1}{2}{\rm i}t\right),
\end{equation}
with
\begin{equation*}
  \begin{split}
    G_1=&45+8\,{t}^{6}-24\,{\rm i}{t}^{4}x-144\,{\rm i}{t}^{3}{x}^{2}-48\,{\rm i}{t}^{2}{x}^{3}-
         72\,{\rm i}t{x}^{4}+576\,{\rm i}{t}^{2}x+288\,{\rm i}t{x}^{2}\\&-144\,{t}^{3}x-504\,{t}^{2}
         {x}^{2}-144\,t{x}^{3}-24\,{\rm i}{x}^{5}-528\,{\rm i}{t}^{3}+414\,{\rm i}t+90\,{\rm i}x-72\,{\rm i}{
         t}^{5}\\
         &+48\,{\rm i}{x}^{3}+24\,{t}^{2}{x}^{4}+24\,{t}^{4}{x}^{2}+504\,tx
         -198\,{x}^{2}-486\,{t}^{2}-60\,{x}^{4}-60\,{t}^{4}+8\,{x}^{6},\\
    G_2=&9+48\,{t}^{3}x-216\,{t}^{2}{x}^{2}+48\,t{x}^{3}+24\,{\rm i}{t}^{5}-24\,{\rm i}{x}^
         {5}+24\,{\rm i}t{x}^{4}+48\,{\rm i}{t}^{3}{x}^{2}+198\,{\rm i}t\\&-24\,{\rm i}{t}^{4}x+180\,{t}^{
         4}+24\,{t}^{2}{x}^{4}+24\,{t}^{4}{x}^{2}+8\,{t}^{6}+8\,{x}^{6}-12\,{x}
         ^{4}+666\,{t}^{2}+90\,{x}^{2}\\&-72\,tx+336\,i{t}^{3}-48\,{\rm i}{x}^{3}-54\,{\rm i}x+288\,{\rm i}{t}^{2}x-48\,{\rm i}{t}^{2}{x}^{3}.
  \end{split}
\end{equation*}

Besides, we can also obtain the $k$-th $(k=3,4,5,6,7)$ order rogue
wave solution of the GI equation. Since their analysis expressions
are too cumbersome,  we omit them. Their dynamical evolutions are
displayed in Fig. \ref{fig.nrw}. From the figures, we find these
local peaks all own a high amplitude on their center, and there are
many small peaks locating around the central peak. Compared with the
higher order rogue wave of the NLS equation, the central high peak
of higher order rogue wave of  the GI equation looks like higher
than the same order of the NLS equation. However, through detailed
analysis, we find the amplitude of the $k$-th order rogue wave
solution of the GI equation and the NLS equation  are both $2(k+1)c$
($c$ is the boundary condition of the
seed solution)\cite{arxiv12093742}.
\end{itemize}



\section{{\bf The dynamics of Rogue Wave}}
In above section, we  obtain the fundamental pattern of the higher
order rogue wave solution for the GI equation with $D_1=1$ and
$D_2=1$ in  (\ref{efnz}). Actually, $D_{1}$ and $D_{2}$ are
arbitrary constant (or go to constant). In this section, we set
$D_1$ and $D_2$ as following:
\begin{equation}
\left\{
\begin{aligned}
D_1=&\exp(-{\rm i}c_1(S_0+S_1\epsilon+S_2\epsilon^2+S_3\epsilon^3+\cdots+S_{k-1}\epsilon^{k-1})),\\
D_2=&\exp({\rm i}c_1(S_0+S_1\epsilon+S_2\epsilon^2+S_3\epsilon^3+\cdots+S_{k-1}\epsilon^{k-1})).
\end{aligned}
\right.
\end{equation}
Here $S_1, S_2,\ldots,S_{k-1}$ $\in$ $\mathbb{C}$. Although the terms with nonzero orders of $\epsilon$ vanish in the $\epsilon\rightarrow0$ limit, analysis and numerics prove that their coefficients
$S_i(i=1,2,\ldots,k-1)$ have a crucial effect on the dynamics of higher order rogue wave. In the following, our main task is to talk about how do these parameters control these different spatial-temporal structures at the same order $k$.\\

{\bf 4.1. Solutions with one parameter}\\

\begin{itemize}
\item $S_0$: Fundamental pattern

When $S_i=0$ $(i\neq0)$, we obtain a trivial translation. A special case for the first order rogue wave is shown in Fig. \ref{fig.1rw1}. It just change the location of the rogue wave. Actually we can shift the location of rogue wave to arbitrary position by hanging the value of $S_0$. The case for the NLS equation had been given in\cite{Communtheorphys56631}. Moveover, the $k$ order rogue waves have $n(n+1)/2$ waves. Starting from large negative t, n small peaks, then a row of $n-1$ larger peaks etc., the central high amplitude wave appears. The process is reversed in positive t. The dynamical evolution can be observed for the NLS equation distinctly in \cite{PRE84056611}.


\item $S_1$: Two triangular structure
\begin{enumerate}
\item Triangular structure

Let all each coefficient $S_i=0$ except $S_1$. The higher
order rogue wave solution of the GI equation split into a triangular structure. This structure of $k$-th order rogue wave contains $k(k+1)/2$  first order fundamental patterns, which make up $k$  successive arrays, and these arrays possess $k, k-1,\cdots,1$ peaks respectively.  These structures of the $k$-th $(k=2,5,7)$ order rogue wave are shown in Fig. \ref{fig.rwtri}. From the Fig. \ref{fig.rwtri}(c), we observe seven arrays, and each of them has $7,6,5,4,3,2,1$ local maxima respectally. And the orientation of these triangular structures in the $(x,t)$-plane remain the same. Our result is a general  case of \cite{PLA3752782}. \\

\item Modified triangular structure

Actually, the outer triangular is independent of the inner triangular in above triangular structure. For instance, when $k=5$, the outer $12$ local maxima and the inner $3$ fundamental patterns both make up a triangular, but these tow triangles are irrelevant. If we alter the appearance of (\ref{efnz}) and mixing
coefficients $D_1$ and $D_2$ as following:
\begin{equation}\label{mix2}
\left(\mbox{\hspace{-0.2cm}}\begin{array}{c}
 \phi_{k}(\lambda_{k})\\
 \varphi_{k}(\lambda_{k})\\
\end{array}\mbox{\hspace{-0.2cm}}\right)=\left(\mbox{\hspace{-0.2cm}}
\begin{array}{c}
 D_1\varpi_1(\lambda_{k})[1,k]+D_1\varpi_2(\lambda_{k})[1,k]-D_2\varpi_1^{\ast}({\lambda_{k}^{\ast}})[2,k]-D_2\varpi_2^{\ast}({\lambda_{k}^{\ast}})[2,k]\\
 D_1\varpi_1(\lambda_{k})[2,k]+D_1\varpi_2(\lambda_{k})[2,k]+D_2\varpi_1^{\ast}({\lambda_{k}^{\ast}})[1,k]+D_2\varpi_2^{\ast}({\lambda_{k}^{\ast}})[1,k]\\
\end{array}
\mbox{\hspace{-0.2cm}}\right),
\end{equation}
with
\begin{equation}
\left\{
\begin{aligned}
D_1=&\exp(-{\rm i}c_1^2(S_0+S_1\epsilon+S_2\epsilon^2+S_3\epsilon^3+\cdots+S_{k-1}\epsilon^{k-1})),\\
D_2=&\exp({\rm i}c_1^2(S_0+S_1\epsilon+S_2\epsilon^2+S_3\epsilon^3+\cdots+S_{k-1}\epsilon^{k-1})).
\end{aligned}
\right.
\end{equation}
Remark: the expression of  \eqref{mix2} is different with \eqref{efnz}.

The inner three peaks form a second fundamental pattern inversely. It is a new structure which has never been obtained for the NLS equation.
Its dynamical evolution is shown in Fig. \ref{fig.rwtri1}.  For $k$-th order rogue wave, we conjecture that there are $3k-3$ first order rogue wave locating the outer triangular shell and an  $(k-3)$-th order rogue wave locating in the center.  \\
\end{enumerate}

\item $S_{k-1}$: Ring structure

If $S_i=0$ $(i\neq k-1)$, we obtain a ring structure. A high maximum peak is surrounded with some local maxima. By simple analysis, we find the peaks locating on the outer shell are all first order rogue waves,
and the number of them is $2k-1$. Besides,
the inner peak is a higher order rogue wave (except $k=3$), whose order is $k-2$ (when $k=3$, the inner peak is a first order rogue wave). For example, when $k=7$, there are $13$ first order rogue wave locating on the outer shell, and a fifth order rogue waves locates in the center. Some of these structures are shown in
Fig. \ref{fig.ring}.\\
\begin{enumerate}
  \item $S_{k-1}$, $S_1$: Ring-triangle

  If $S_1$ is also non-zero, the inner higher order rogue wave will become a triangular structure. For instance, when $k=5$, the inner second order rogue wave is split into a triangle with $S_4=5\times10^9$, $S_1= 500$. Besides, when $k=6,7$, the similar structure is obtained.  Therefore, we can conclude reasonably that the inner higher order rogue wave is able to split into  triangular structure when $S_{k-1}$ is big enough and $S_1\neq0$. We display some of these special models in Fig. \ref{fig.ringtriangle}. Notably, the orientation of the triangle in $7$-th ring-triangle model is different with the others. \\

  \item $S_{k-1}$, $S_{k-3}$: Multi-ring

  Being similar to the above case, the inner higher order rogue wave can be split into a ring structure too. Keeping $S_{k-1}\gg0$ and setting $S_{k-3}\neq0$ ($k>3$), we will obtain a ring structure. In this case, if the central peak is still a higher order rogue wave, we can continue splitting it into first order fundamental
  models. For example, when $k=6$ with $S_5=8\times10^{9}$, $S_3=80000$, we obtain a multi-ring structure with a seconde order rogue wave locating in the center. Its dynamical evolution is shown in Fig. \ref{fig.multiring1}.
  Under this circumstance, assuming $S_1\neq0$, then the second order rogue wave is split into a triangular structure. Its evolution is displayed in Fig. \ref{fig.multiring2}. Naturally, these similar structures are also able to  found for $k=7$. Since the third order rogue wave possesses two structures i.e. triangular structure and ring structure, the inner structure of the seventh can be both triangle structure and ring model. Their evolution is shown in Fig. \ref{fig.multiring3}.
\end{enumerate}
\end{itemize}

In general, there are $k-1$ free parameters for $k$-th order rogue
wave solution, which possesses $2^k$ combinations. But there are
only 4 basic models, and by choosing proper parameters, we also get
dynamical evolutions of the 4 basic patterns. Therefore, we can get
a hierarchy of higher order rogue wave of the GI equation with
combination structure of above four basic models.

\section{\bf Conclusion}

In this paper, we modified the generalized Darboux transformation to
get explicit solutions for the GI equation. In the case of vanishing
boundary condition, we obtain soliton, rational traveling soliton,
breather type soliton, and soliton colliding with breather type
soliton. The last two kinds of solutions are new for the GI
equation. Moreover, we get the formula by applying the Taylor
expansion and limit technique, when eigenvalues share the same
value. Under the condition of non-vanishing boundary condition, we
give the formula of the $N$-rogue wave solution for the GI equation
in proposition 1. As applications, we give the expressions of higher
order rogue wave solutions and discuss their structures. Further
more, we generate solutions with different structures by adjusting
the free parameters $D_1$ and $D_2$. In summary, there are four
basic patterns for higher order rogue waves of the GI equation:
fundamental pattern, triangular structure, modified structure and
ring structure. By choosing proper parameter, we can get solutions
with combination structure of basic ones.

Our results clearly show the connection between shift parameters and
structures of rogue wave due to the explicit formula. These types of
solutions are new to the GI equation. Moreover, some basic models
such as the fundamental pattern, ring structure and triangular
structure have been found in higher rogue waves for other equations
such as the NLS equation \cite{PRE80026601}, but the modified
triangular structure is new to the GI equation. Meanwhile, we get a
family of solutions with combination structures. All of these will
help us in founding univeral properties for rogue waves and better
understanding of ``waves appear from nowhere and disappear without a
trice'' \cite{PLA373675}.

{\bf Acknowledgments}
This work is supported by the NSF of China under Grant No.11271210, No.10971109 and K. C. Wong
Magna Fund in Ningbo University. Jingsong He is also supported by Natural Science Foundation of
Ningbo under Grant No. 2011A610179. We want to thank Prof. Yishen Li (USTC, Hefei, China) for
his long time support and useful suggestions.



\newpage

  \begin{figure*}[!t]
  \centering
  \subfigure[]{\includegraphics[height=4cm,width=4cm]{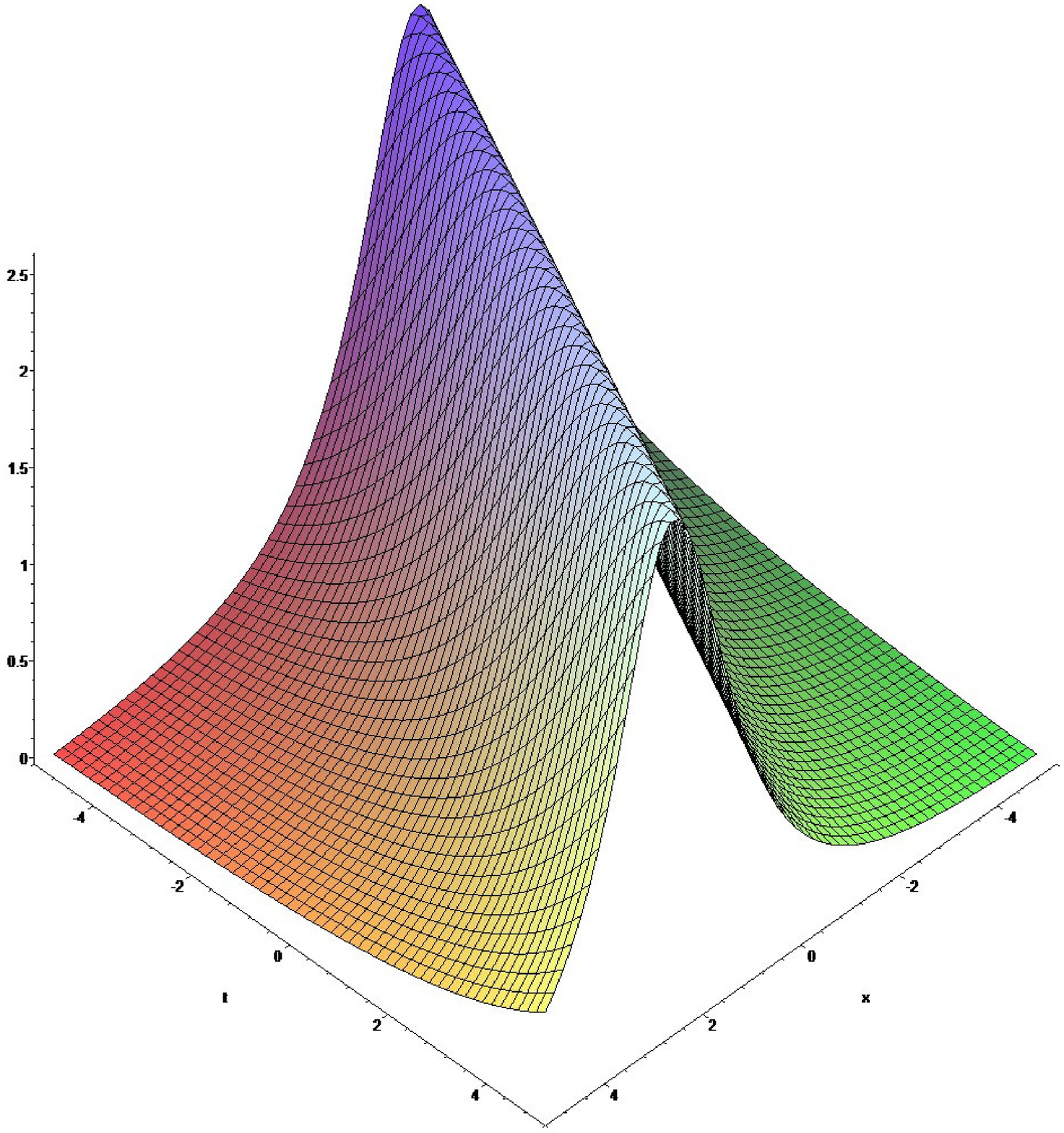}}
  \qquad
  \subfigure[]{\includegraphics[height=4cm,width=4cm]{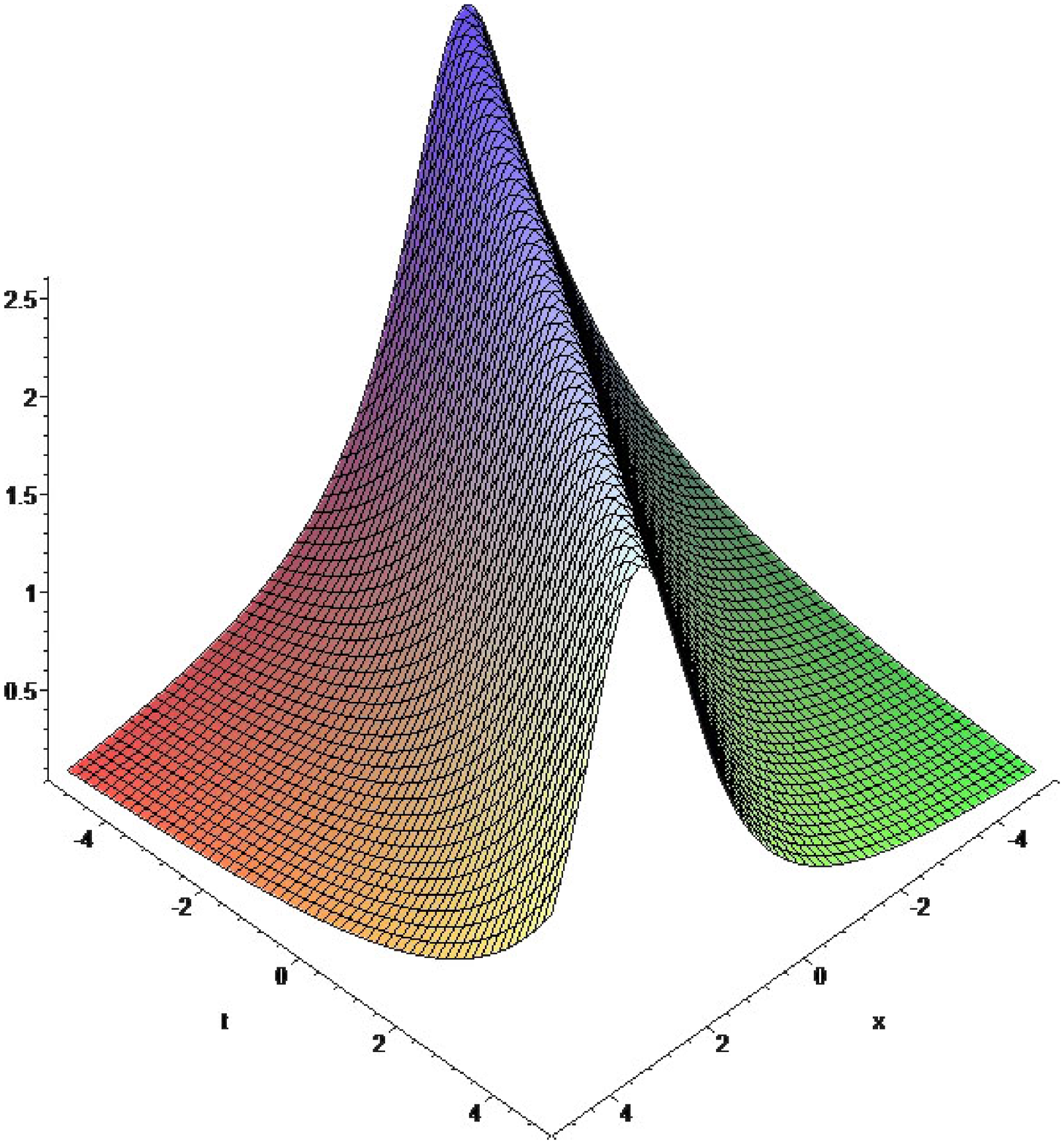}}
  \qquad
  \subfigure[]{\includegraphics[height=4cm,width=4cm]{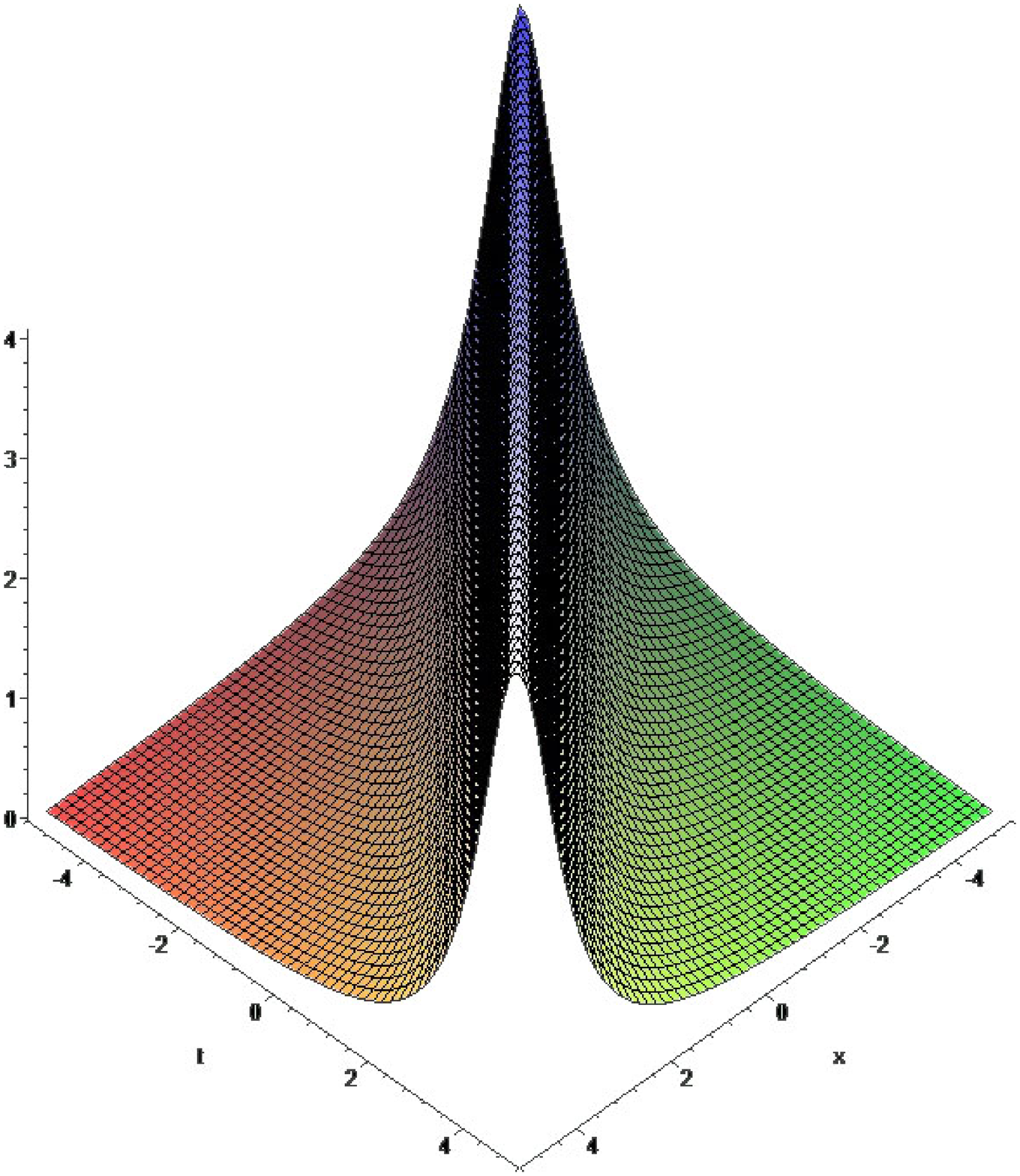}}
  \caption{(Color online) The dynamics of $\left|q^{[2]}\right|^2$. (a)  Eq. (\ref{q[2]}) soliton solution with $\alpha_1=0.2$, $\beta_1=0.4$. (b) Eq. (\ref{q[21]}) rational travelling solution with $\beta_1=0.4$. (c) Eq. (\ref{q[22]}) soliton solution with $\beta_1=0.45$, $\beta_2=0.55$.}\label{fig.soliton}
  \end{figure*}

  \begin{figure*}[!t]
    \centering
    \subfigure[]{\includegraphics[height=4cm,width=4cm]{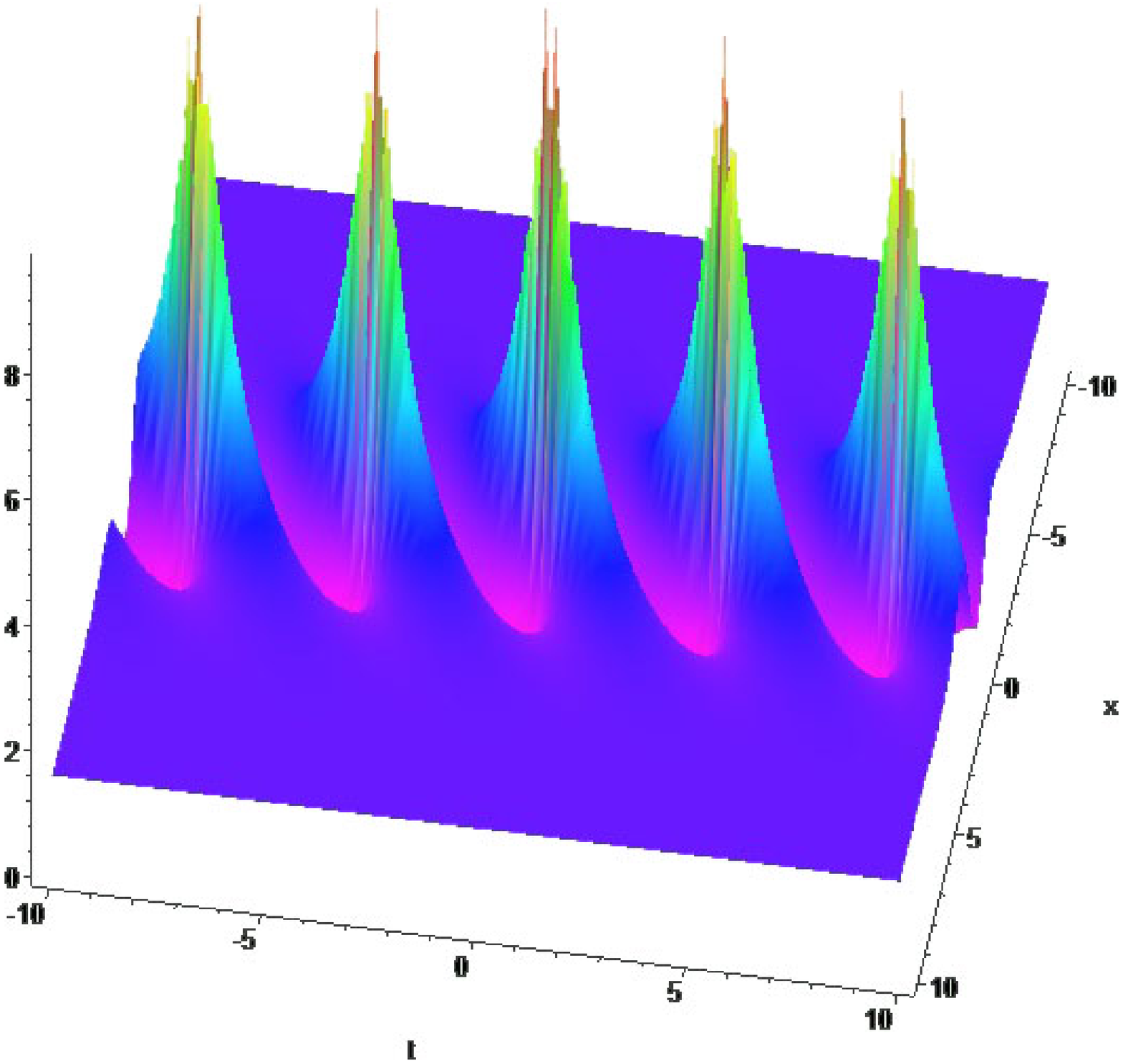}}\qquad
    \subfigure[]{\includegraphics[height=4cm,width=4cm]{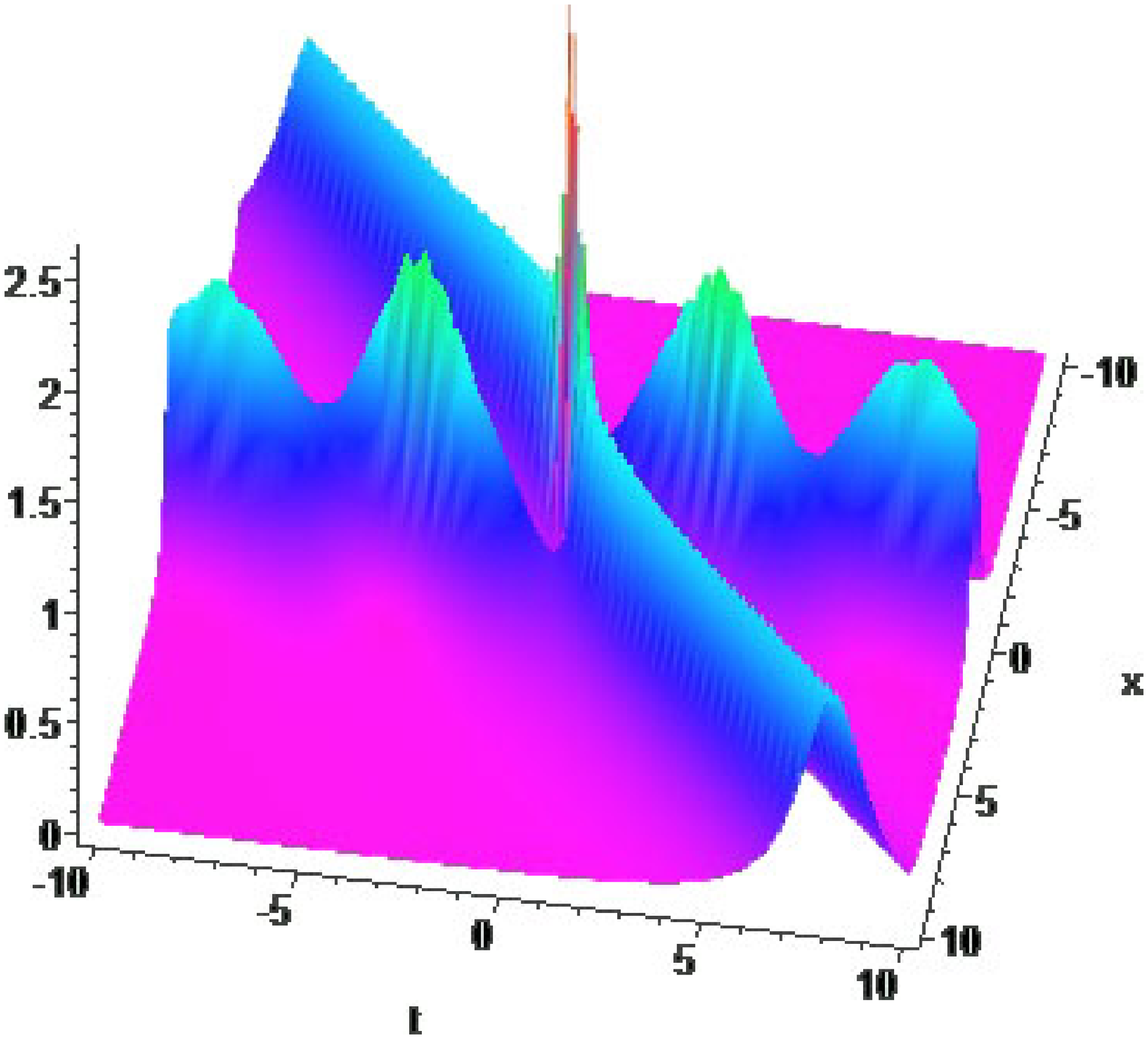}}
    \caption{(Color online) (a) The dynamics of soliton solution quasi-periodic solution on the ($x$, $t$) plane. (b) soliton solution under the quasi-periodic background on the ($x$, $t$) plane.}
   \end{figure*}



\begin{figure*}[!htp]
\centering
\subfigure[]{\includegraphics[height=4cm,width=4cm]{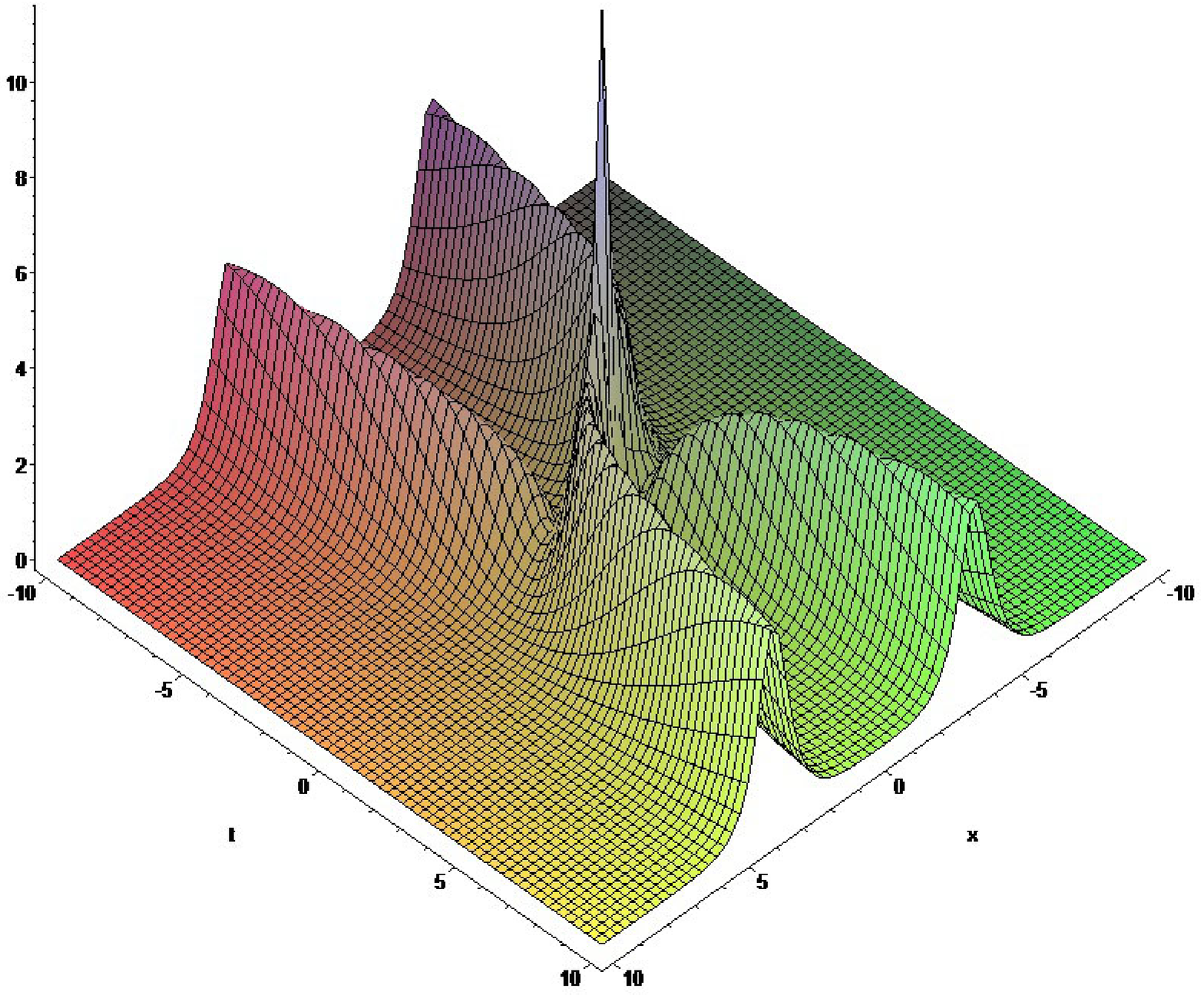}}
\quad
\subfigure[]{\includegraphics[height=4cm,width=4cm]{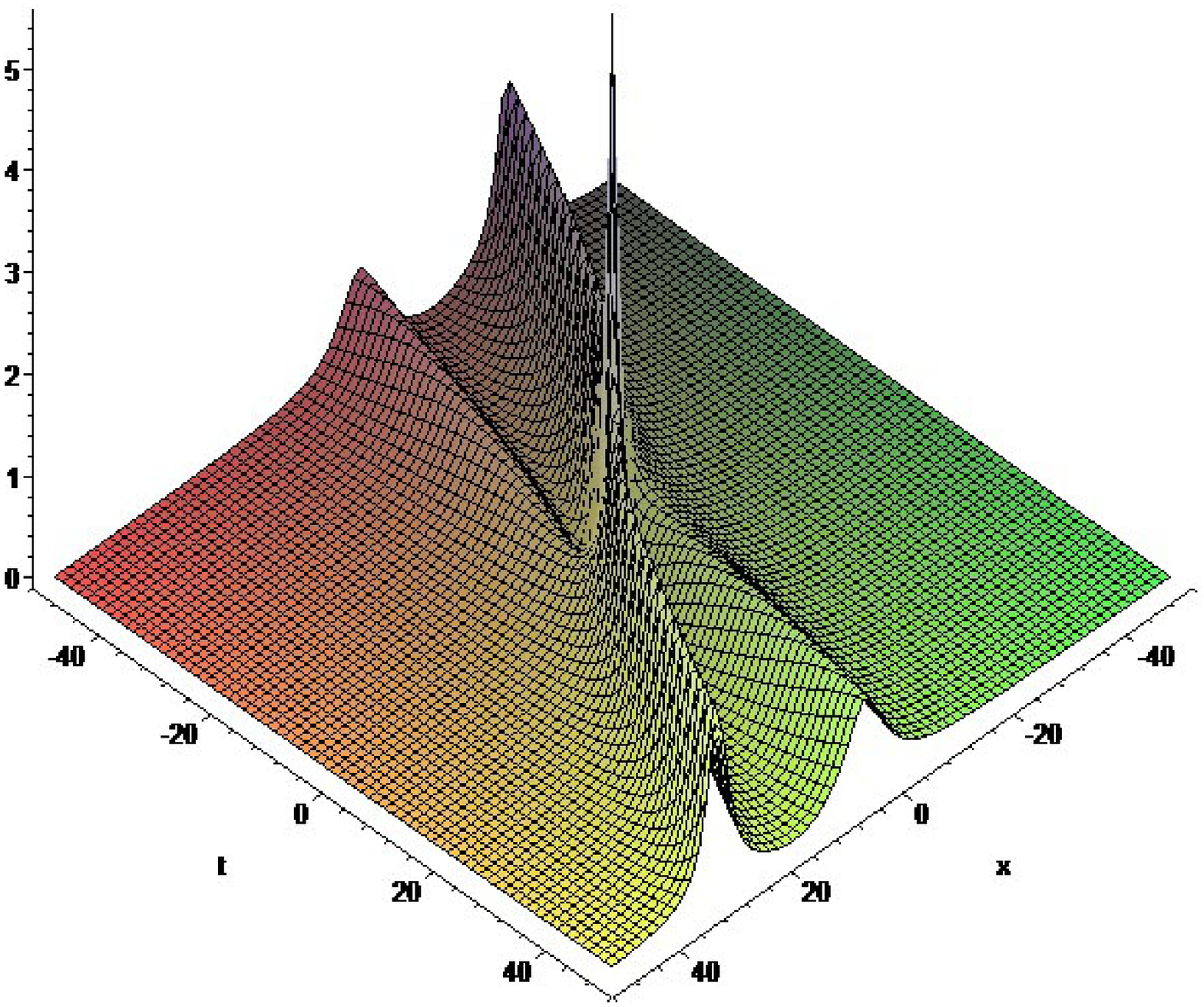}}
\caption{(Color online) The dynamics of $2$-posoliton solution and $2$-rational solution on the (x, t) plane. (a)
Eq. (\ref{2p}) with $\alpha_1=0.5$, $\beta_1=0.5$. (b) Eq. (\ref{2p}) with $\alpha_1\rightarrow0$.}
\end{figure*}

\begin{figure*}[!htp]
\centering
\subfigure{\includegraphics[height=4cm,width=4cm]{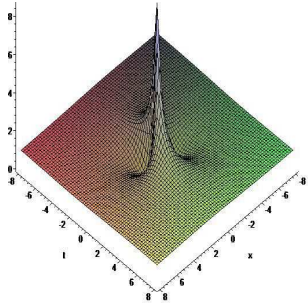}}
\caption{(Color online) The dynamics of the first order rogue wave solution $\left|q^{[2]}\right|^2$ Eq. (\ref{1RW})
on the (x, t) plane  with $a=0$ and $c=1$. The maximum amplitude of
$\left|q^{[2]}\right|^2$ occurs at $t=0$ and $x=0$ and are equal to $9 $.}\label{fig.1rw}
\end{figure*}



\begin{figure*}[!t]
\centering
\subfigure[]{\includegraphics[height=4cm,width=4cm]{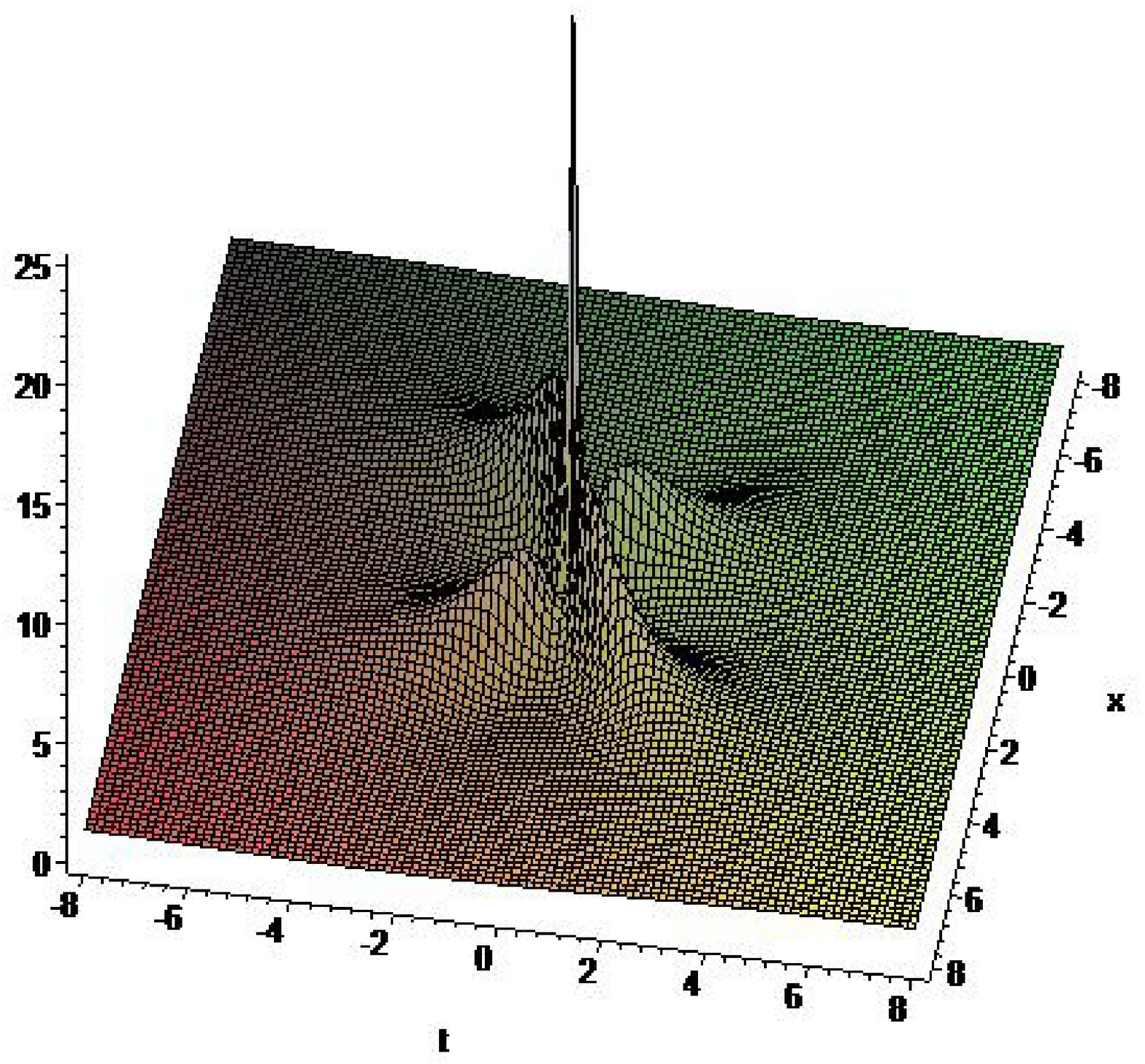}}
\qquad
\qquad
\subfigure[]{\includegraphics[height=4cm,width=4cm]{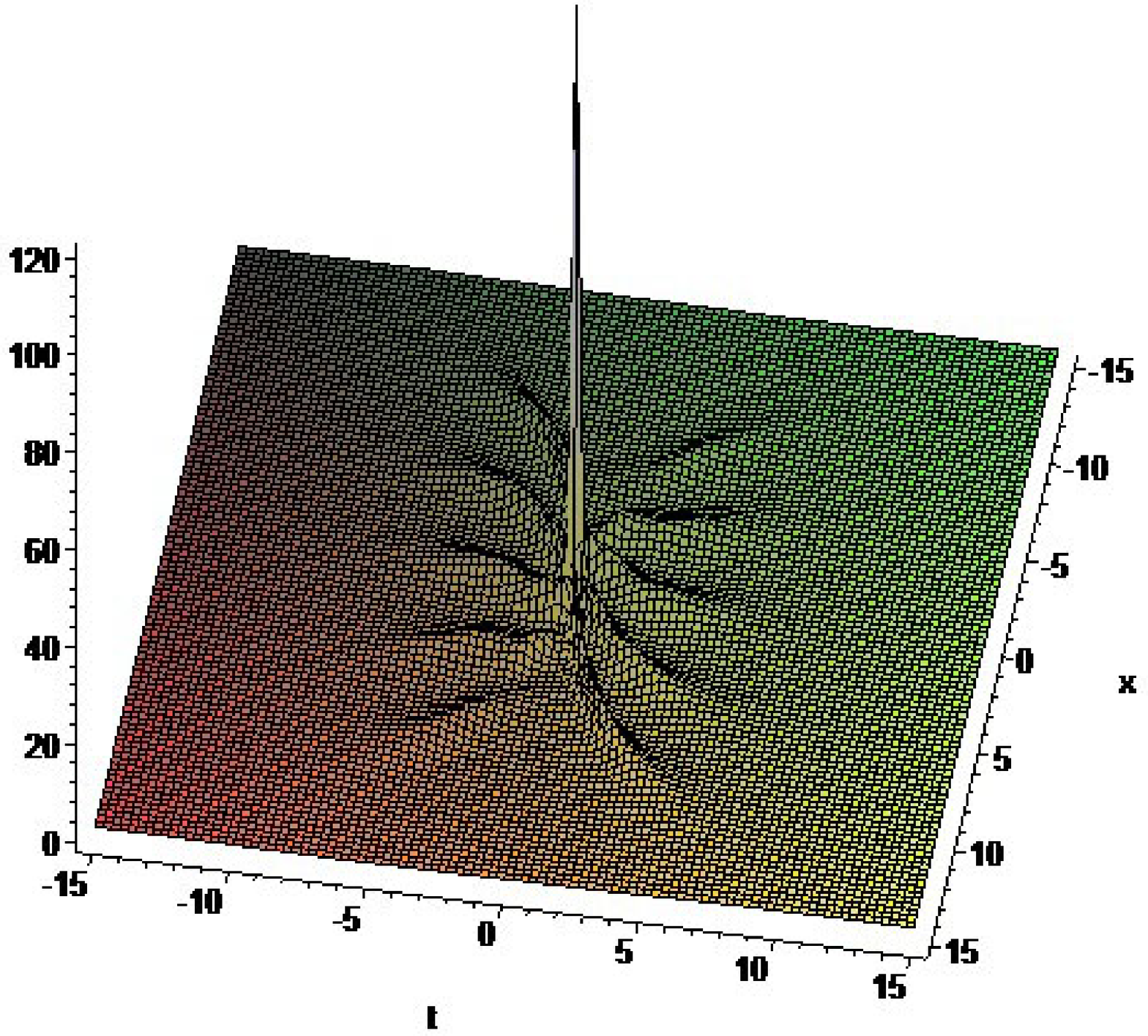}}
\qquad
\qquad
\subfigure[]{\includegraphics[height=4cm,width=4cm]{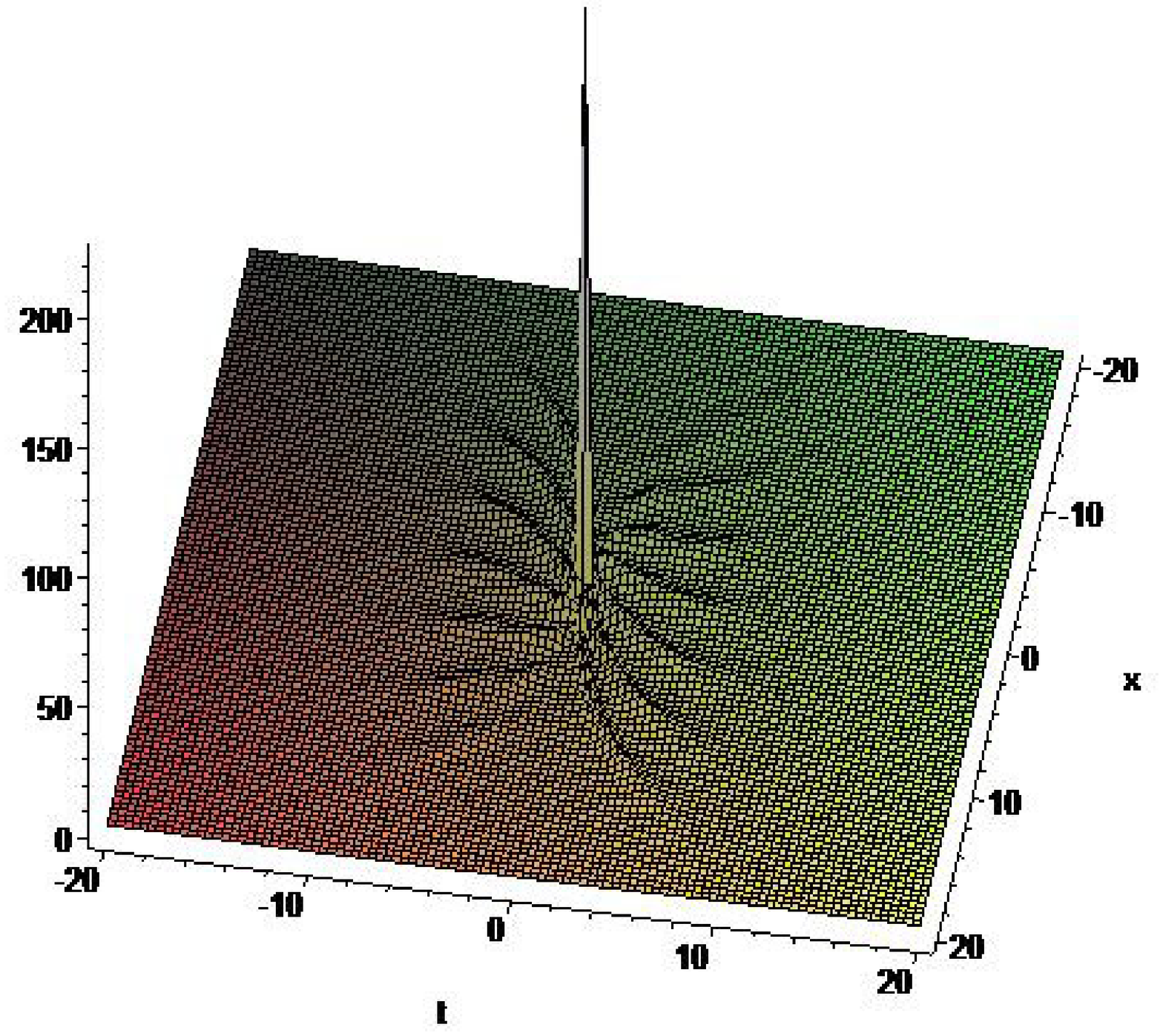}}
\caption{(Color online) The dynamics of higher order rogue wave. (a) The second order rogue wave.
 (b) The fifth order rogue wave.  (f) The seventh order rogue wave.}\label{fig.nrw}
\end{figure*}

\begin{figure*}[!t]
\centering
\subfigure[]{\includegraphics[height=4cm,width=4cm]{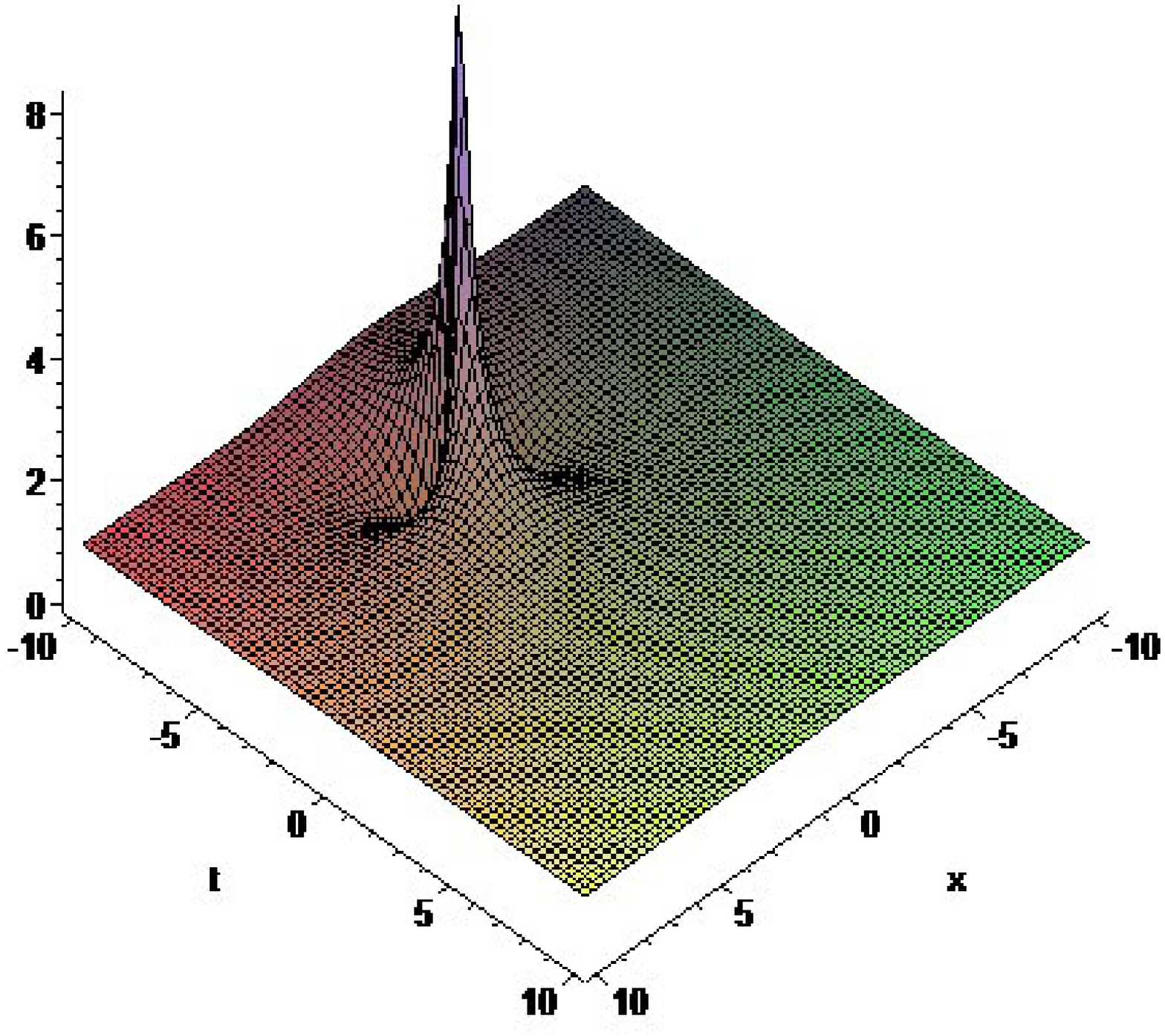}}
\qquad
\subfigure[]{\includegraphics[height=4cm,width=4cm]{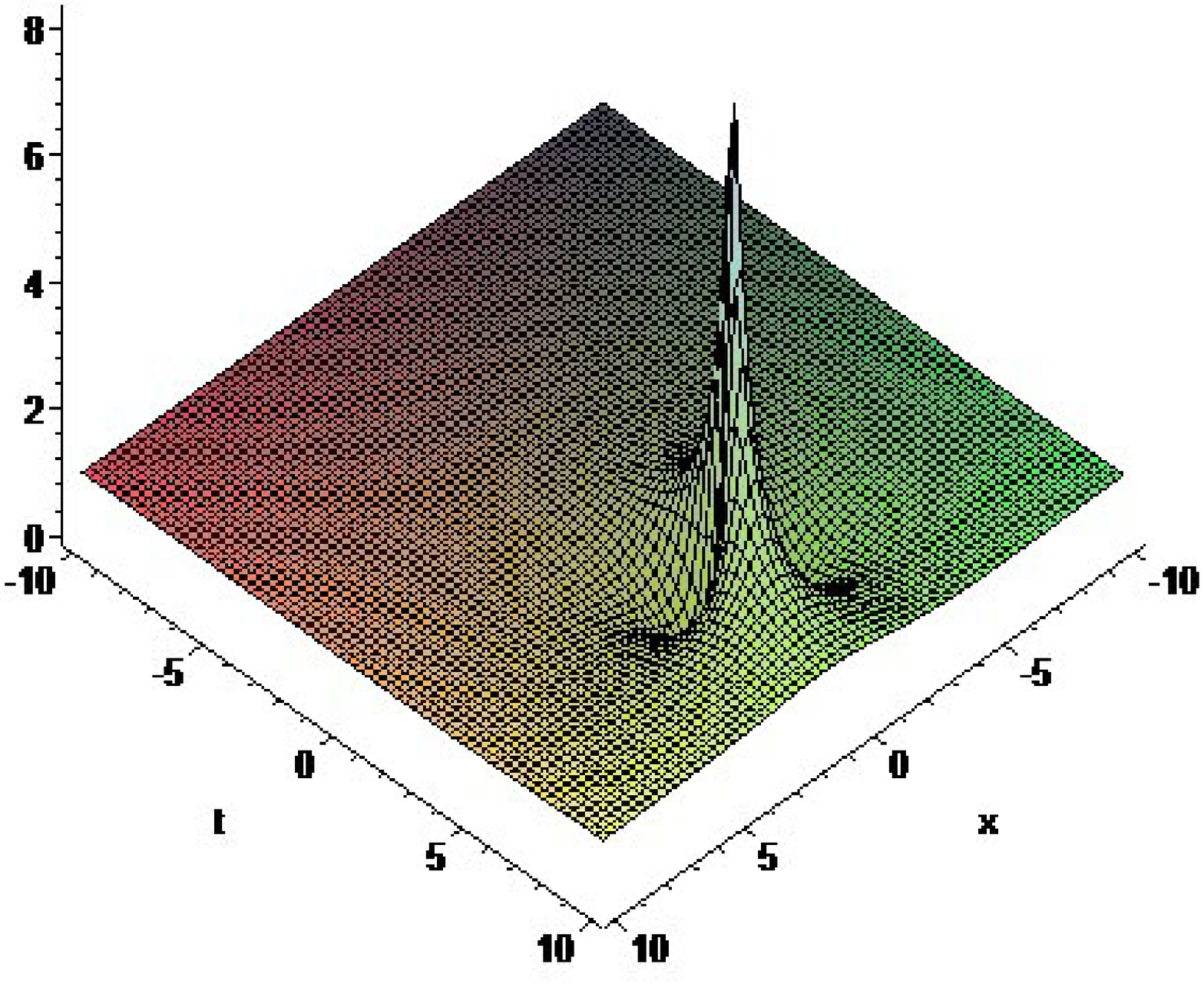}}
\caption{(Color online) The fundamental pattern of first order rogue wave. (a) The first order rogue wave with $S_0=5$, the maximum
 amplitude occurs at $t=-5$ and $x=0$. (b) The first order rogue wave with $S_0=-5$, the maximum amplitude occurs at
 $t=5$ and $x=0$.}\label{fig.1rw1}
\end{figure*}



\begin{figure}[!t]
\centering
\subfigure[]{\includegraphics[height=4cm,width=4cm]{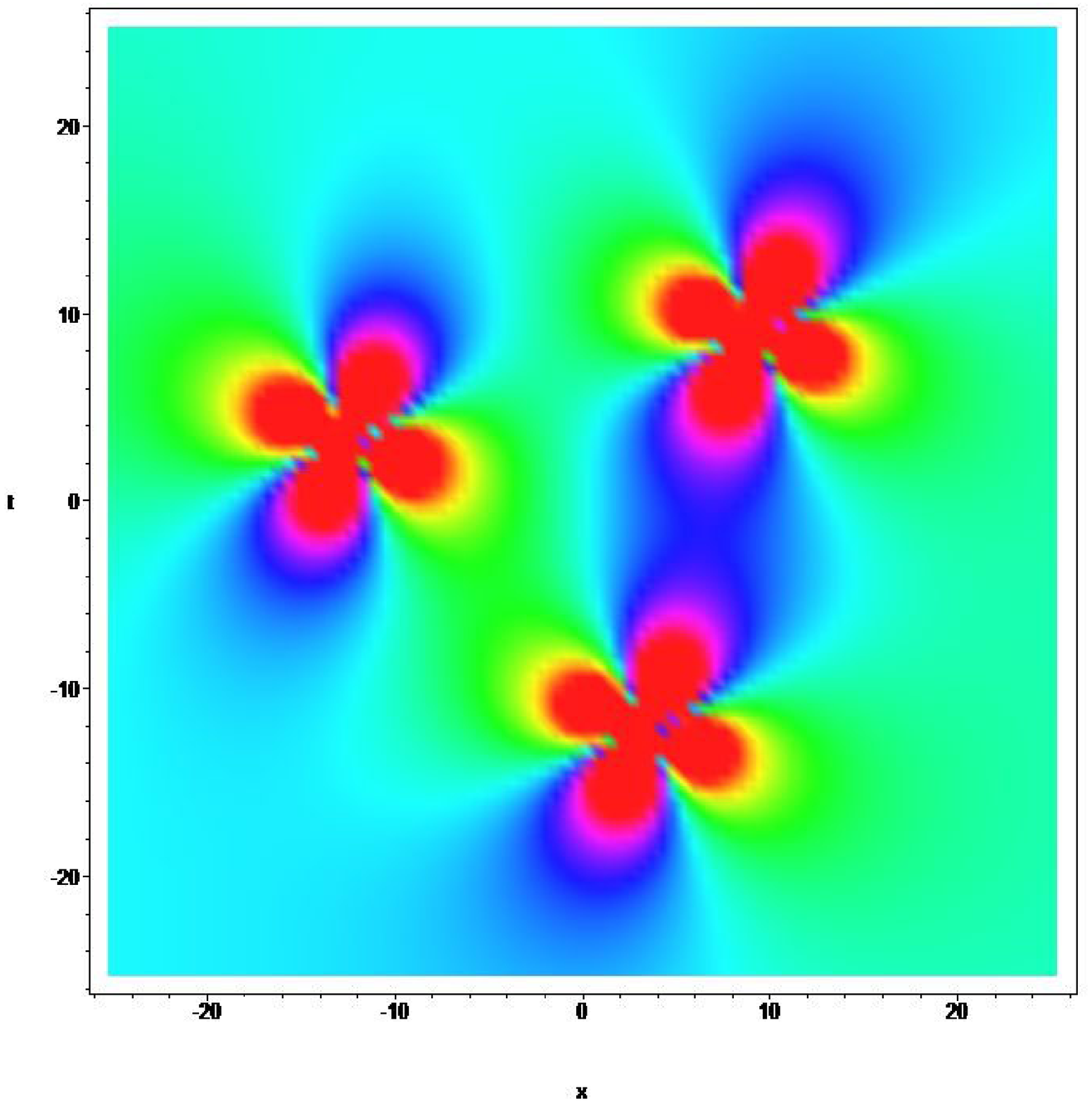}}
\qquad
\subfigure[]{\includegraphics[height=4cm,width=4cm]{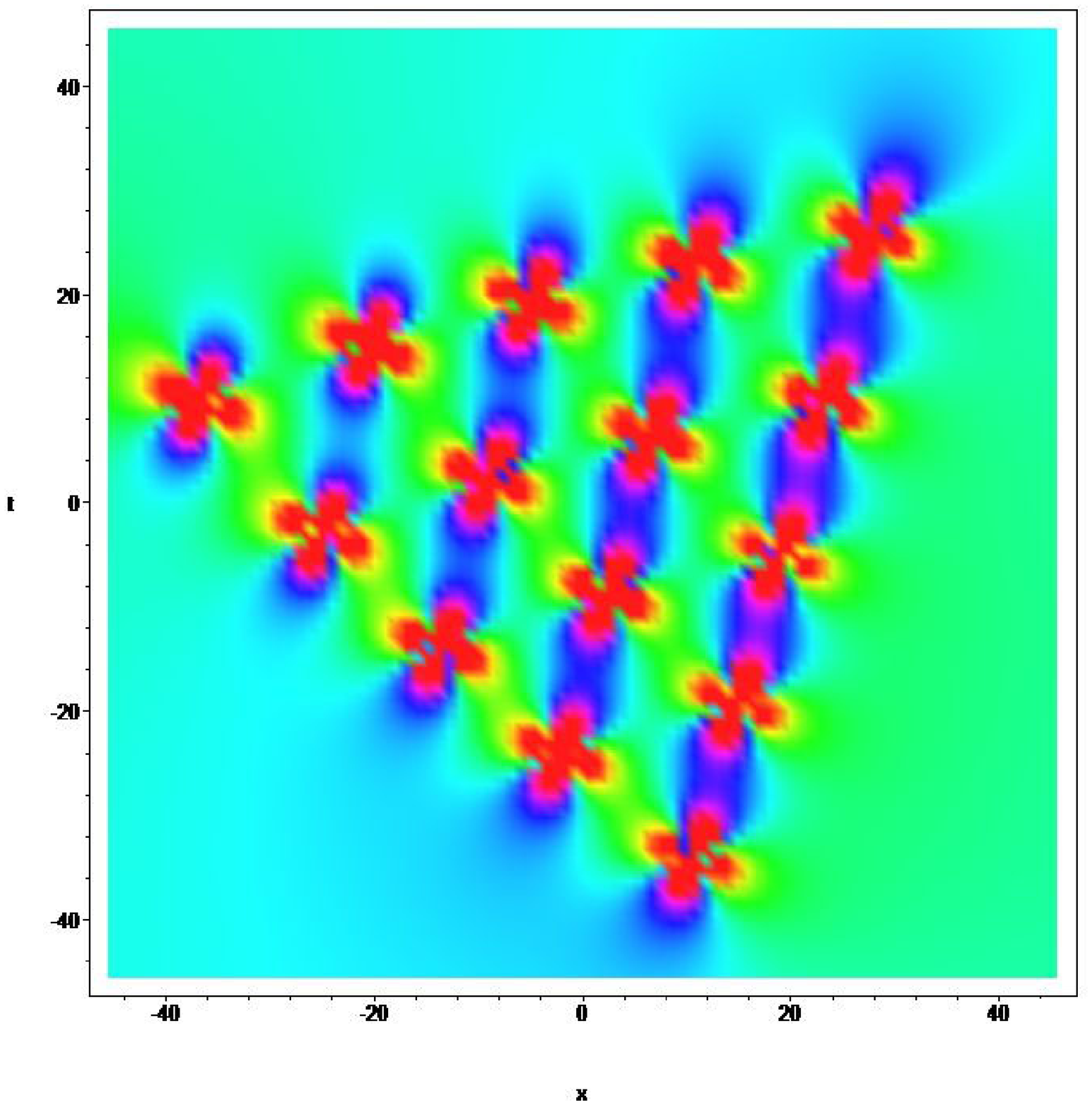}}
\qquad
\subfigure[]{\includegraphics[height=4cm,width=4cm]{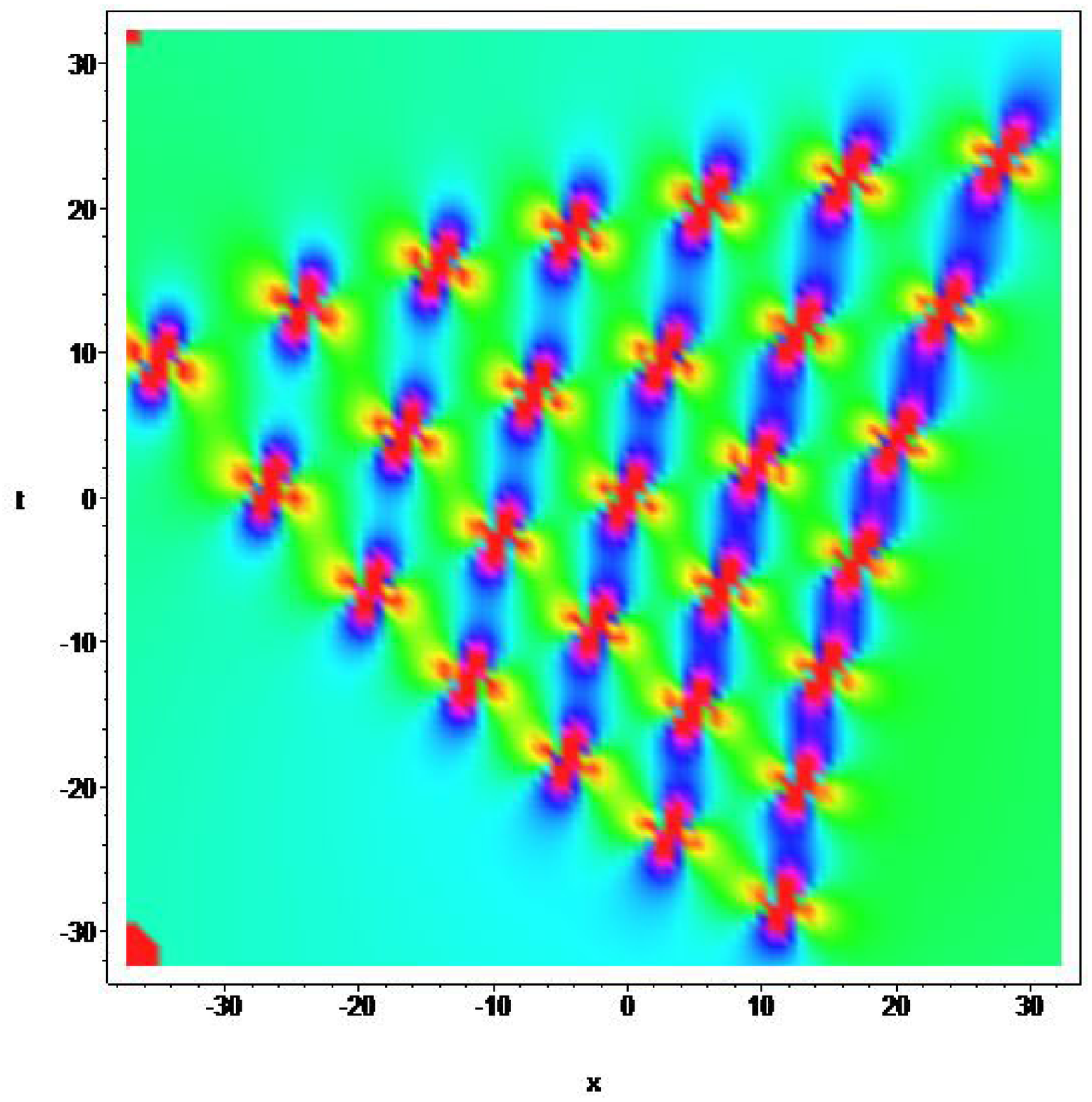}}
\caption{(Color online) The triangle structure of higher order rogue wave. (a) The second order rogue wave with $S_1=1000$. (b) The fifth order rogue wave with $S_1=800$.  (c) The seventh order rogue wave with $S_1=250$.}\label{fig.rwtri}
\end{figure}

\begin{figure}[!t]
\centering
\subfigure[]{\includegraphics[height=4cm,width=4cm]{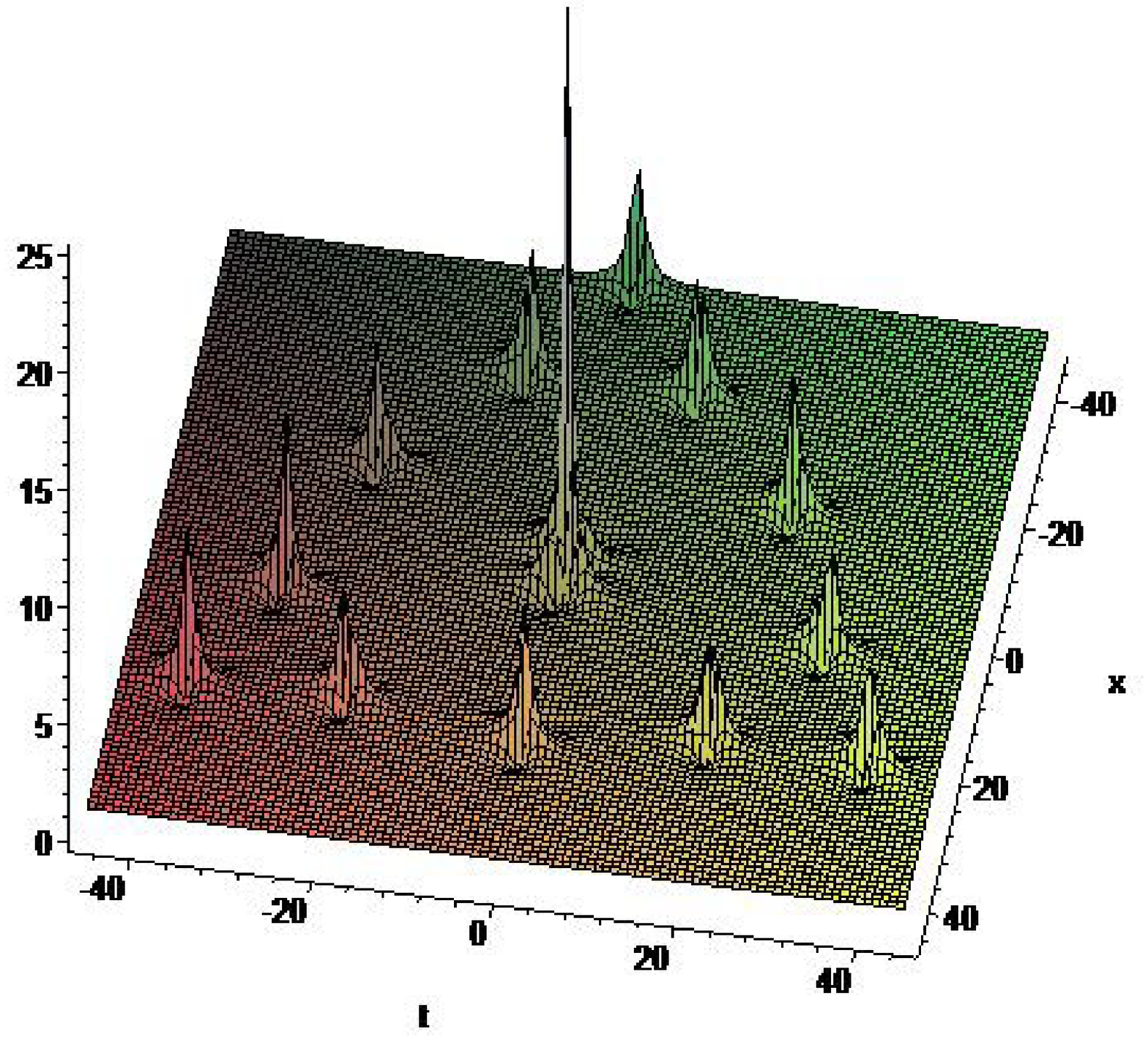}}
\qquad
\subfigure[]{\includegraphics[height=4cm,width=4cm]{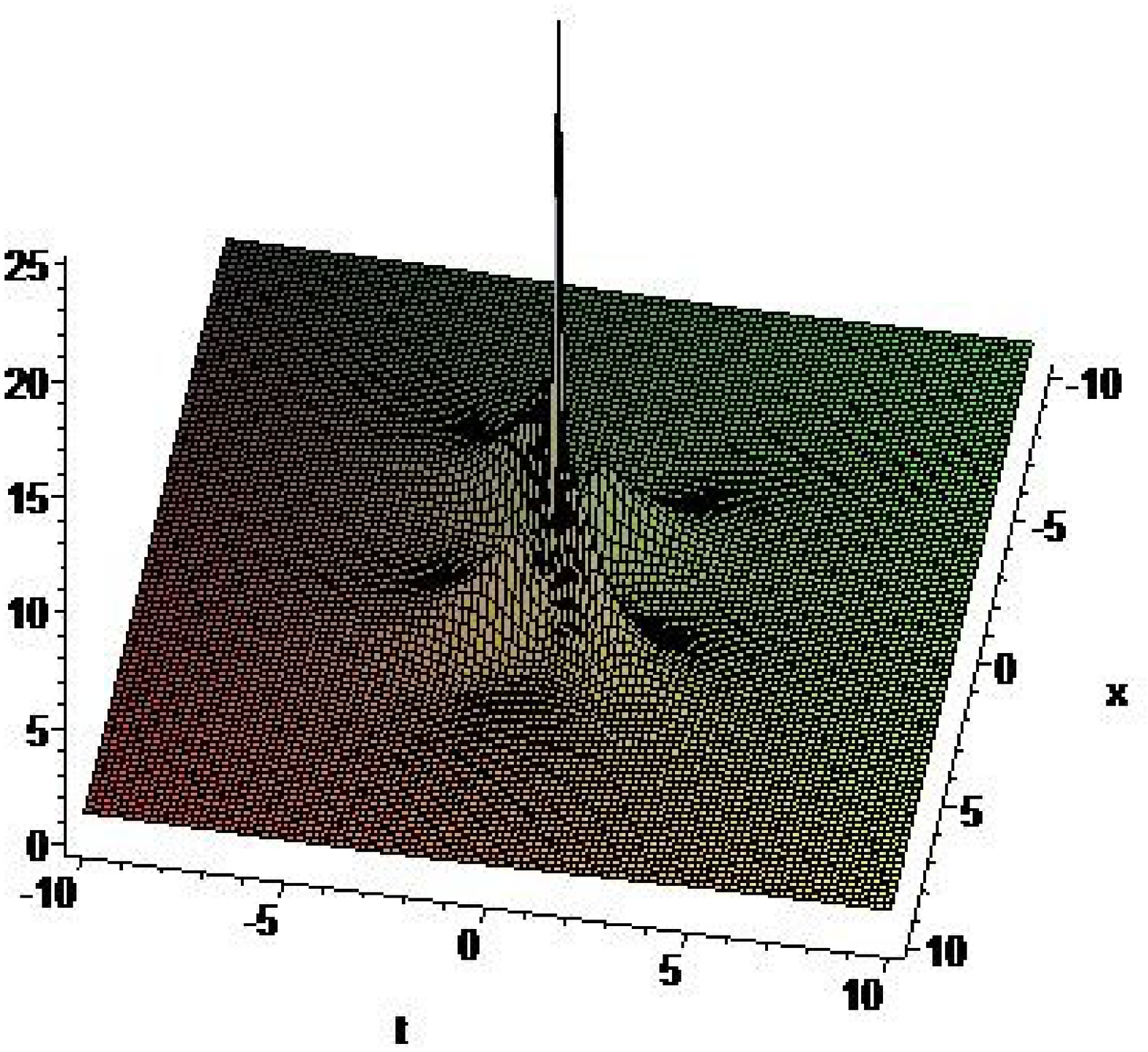}}
\caption{(Color online) The modified triangular structure with a higher order rogue wave locating in the center. (a) An overall profile with $S_1=1000$. (b)A local central profile of the right panel.}\label{fig.rwtri1}
\end{figure}


\begin{figure*}[!htbp]
\centering
\subfigure[]{\includegraphics[height=4cm,width=4cm]{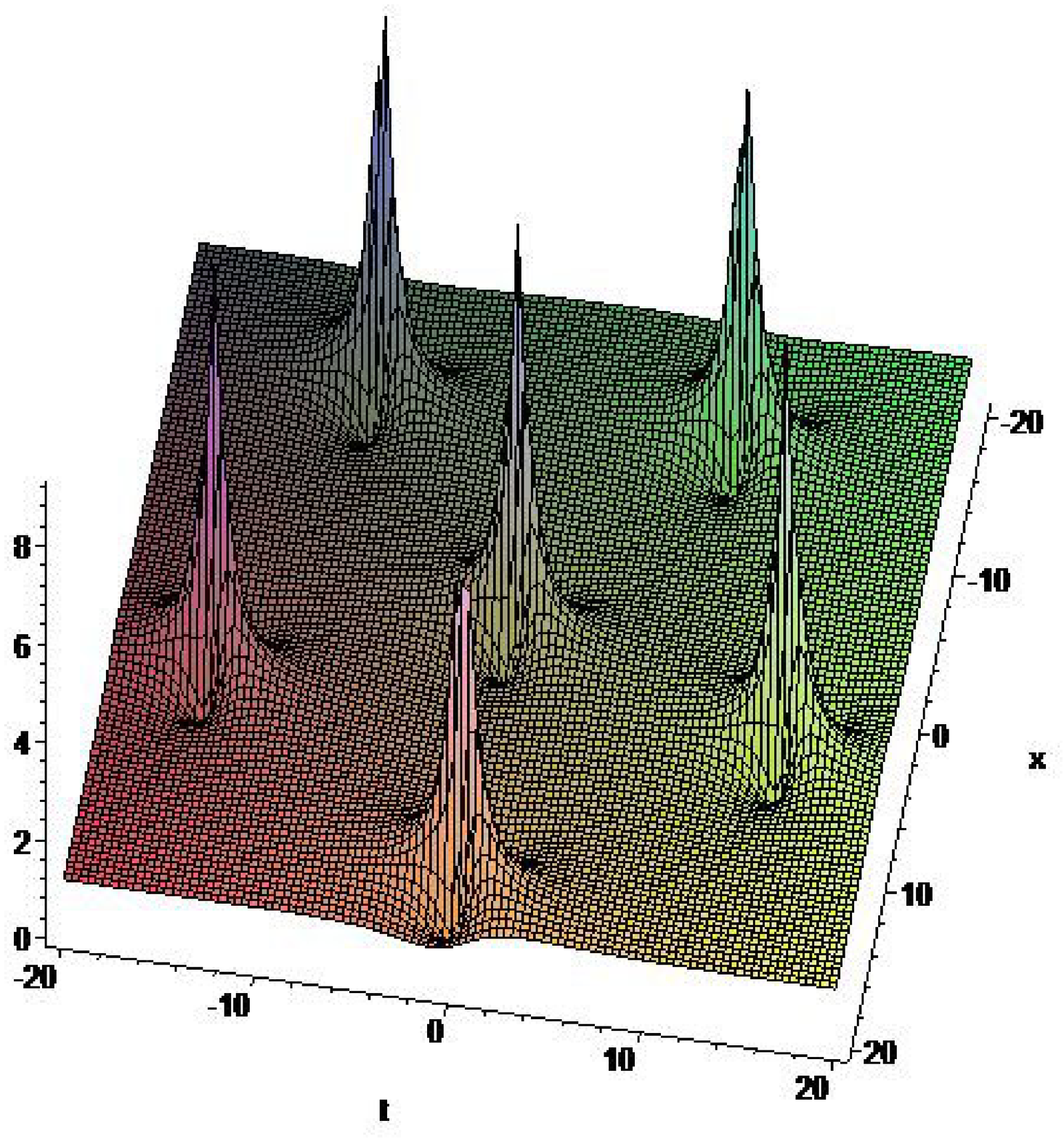}}
\qquad
\qquad
\subfigure[]{\includegraphics[height=4cm,width=4cm]{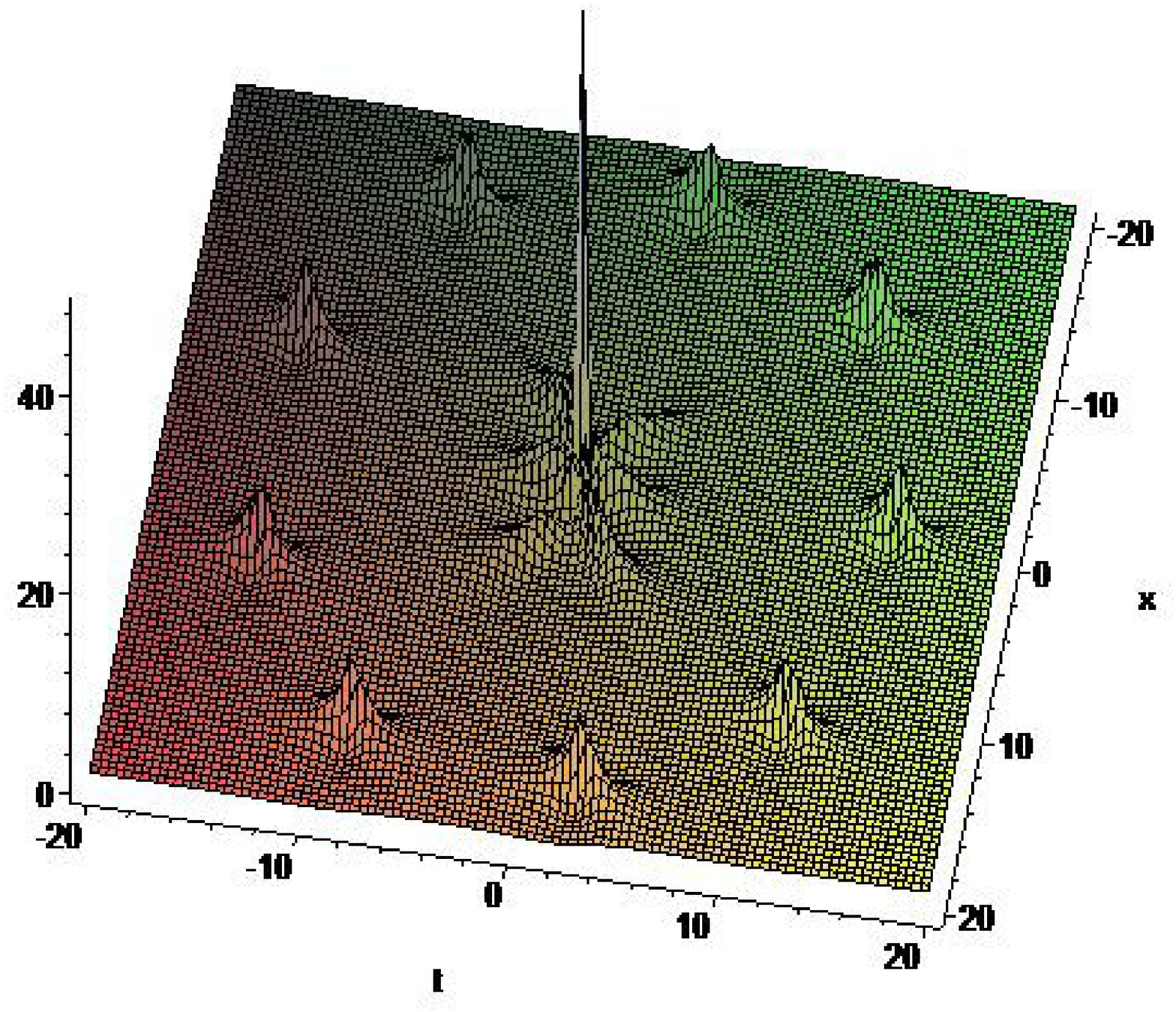}}
\qquad
\qquad
\subfigure[]{\includegraphics[height=4cm,width=4cm]{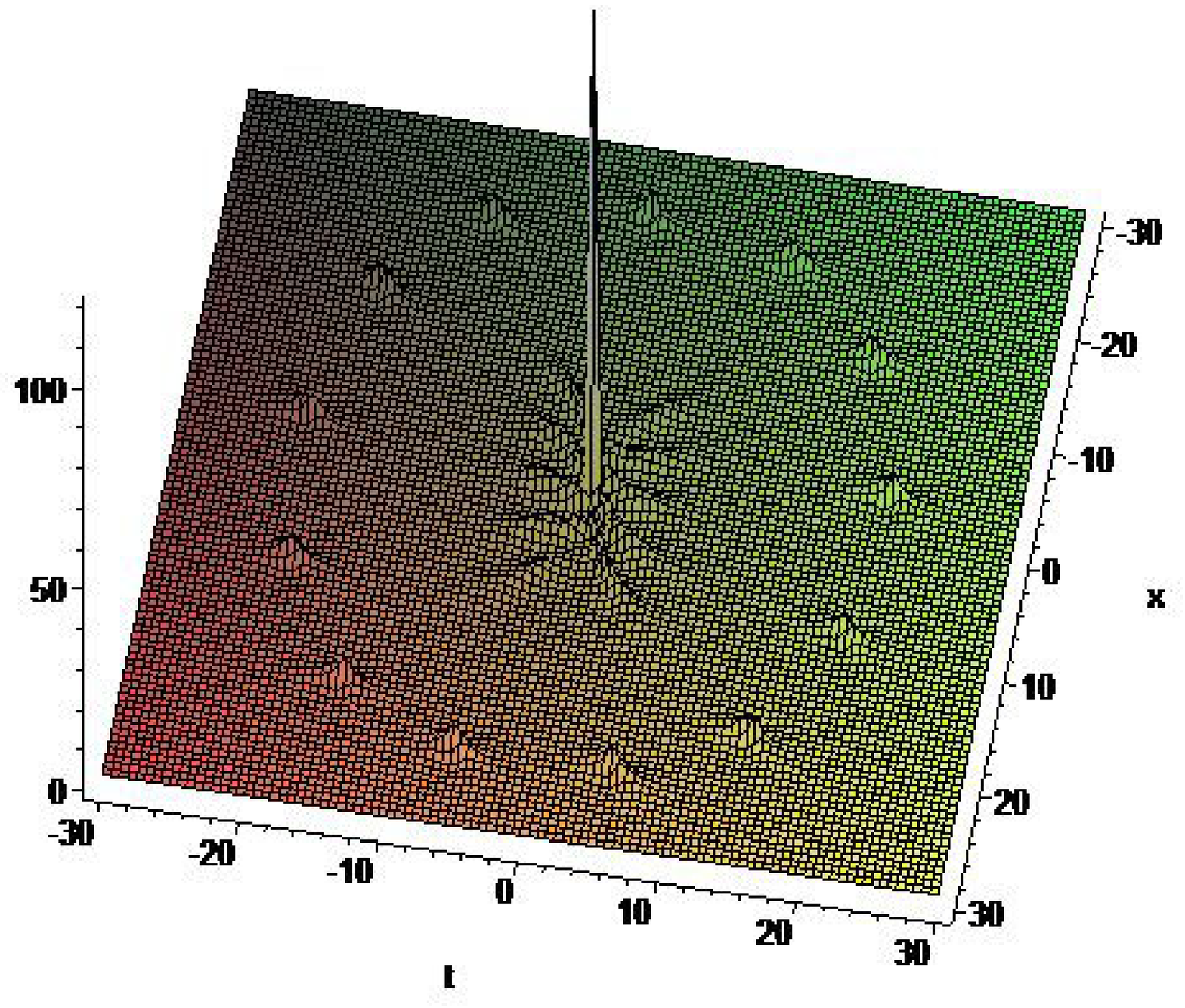}}
\caption{(Color online) The ring structure of higher order rogue wave. (a) The third order rogue wave with $S_2=50000$.  (b) The fifth order rogue wave with $S_4=5000000$.  (c) The seventh order rogue wave with $S_6=5\times10^7$.}\label{fig.ring}
\end{figure*}

\begin{figure*}[!htbp]
\centering
\subfigure[]{\includegraphics[height=4cm,width=4cm]{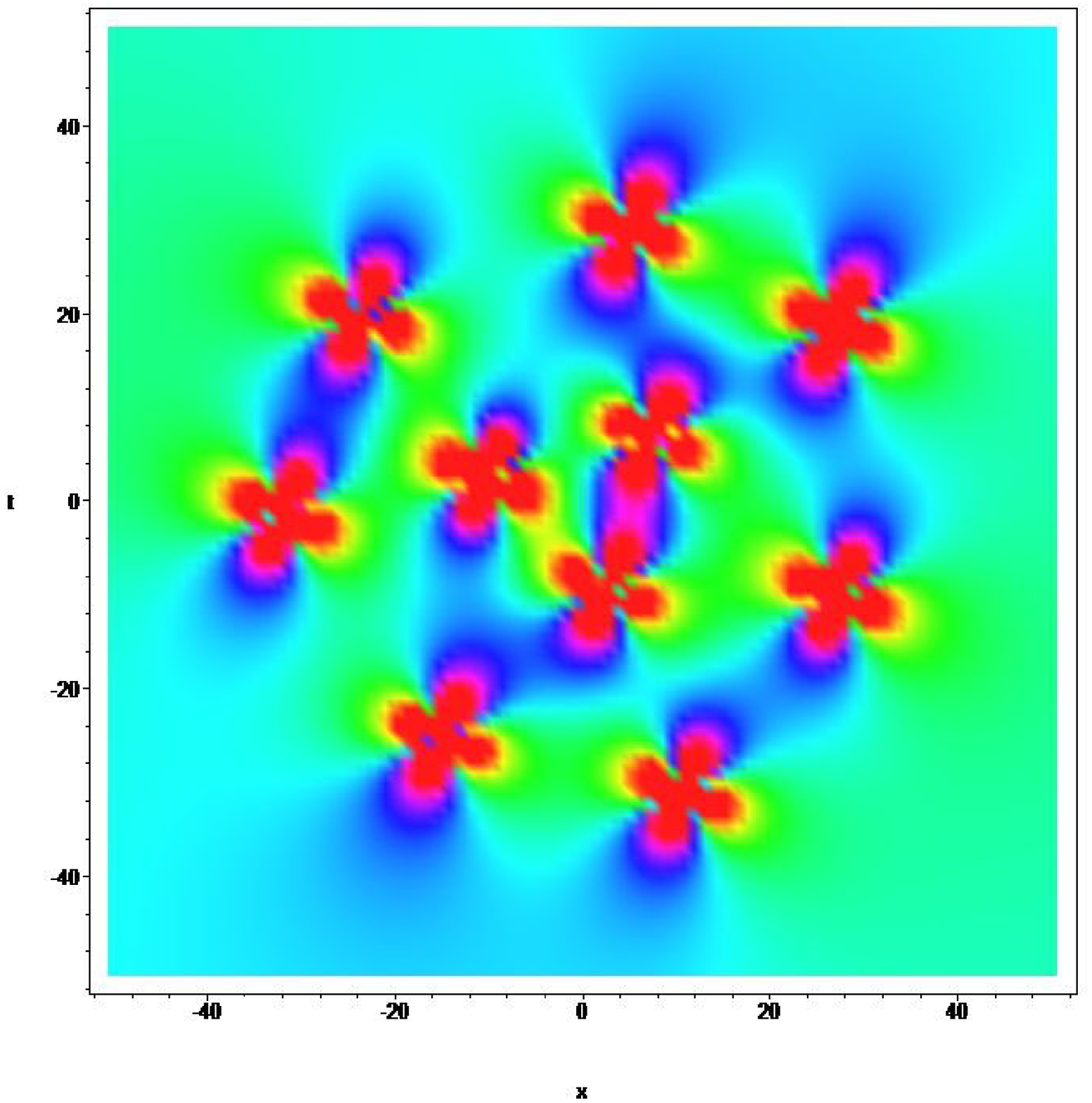}}
\qquad
\subfigure[]{\includegraphics[height=4cm,width=4cm]{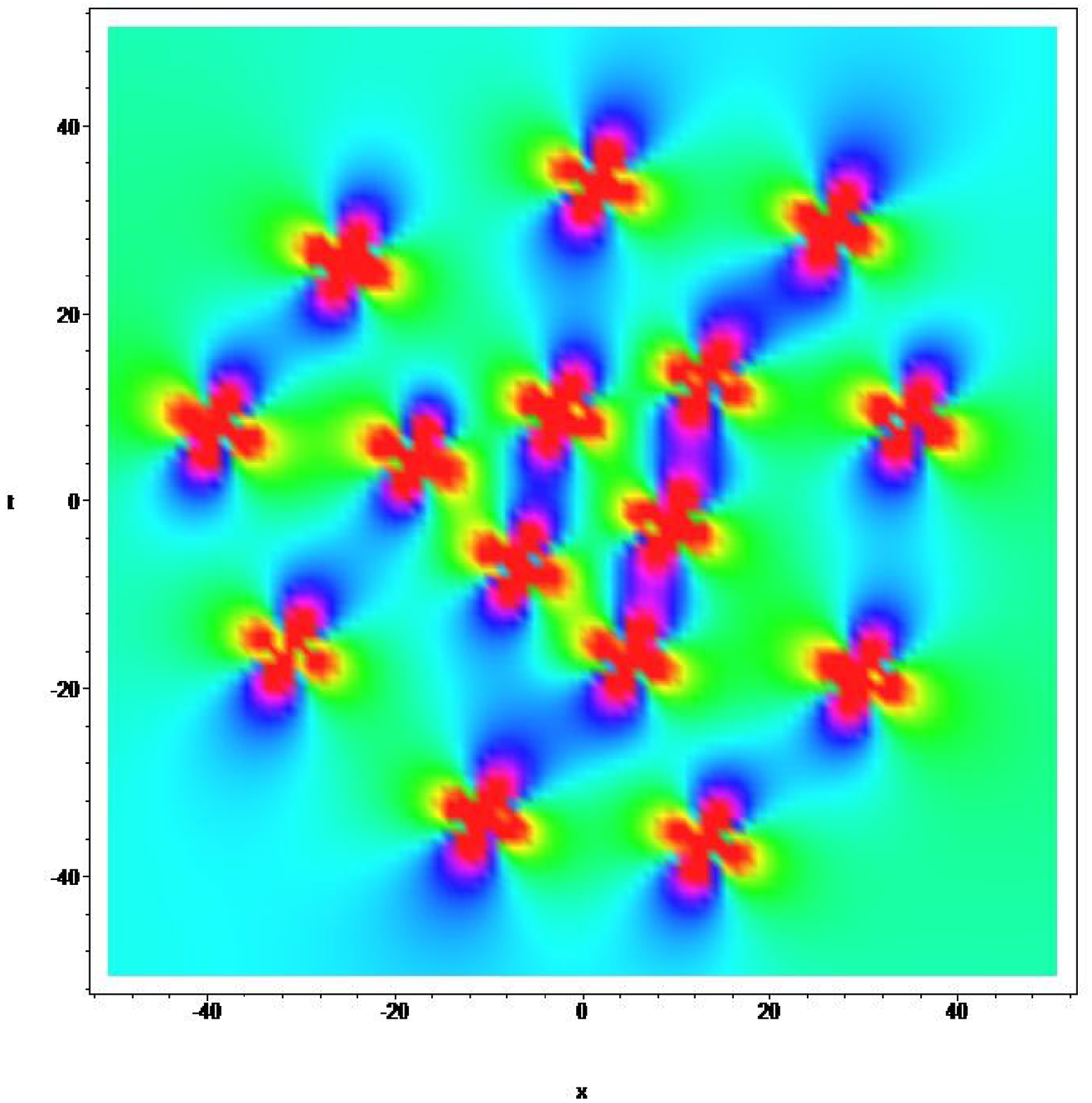}}
\\
\subfigure[]{\includegraphics[height=4cm,width=4cm]{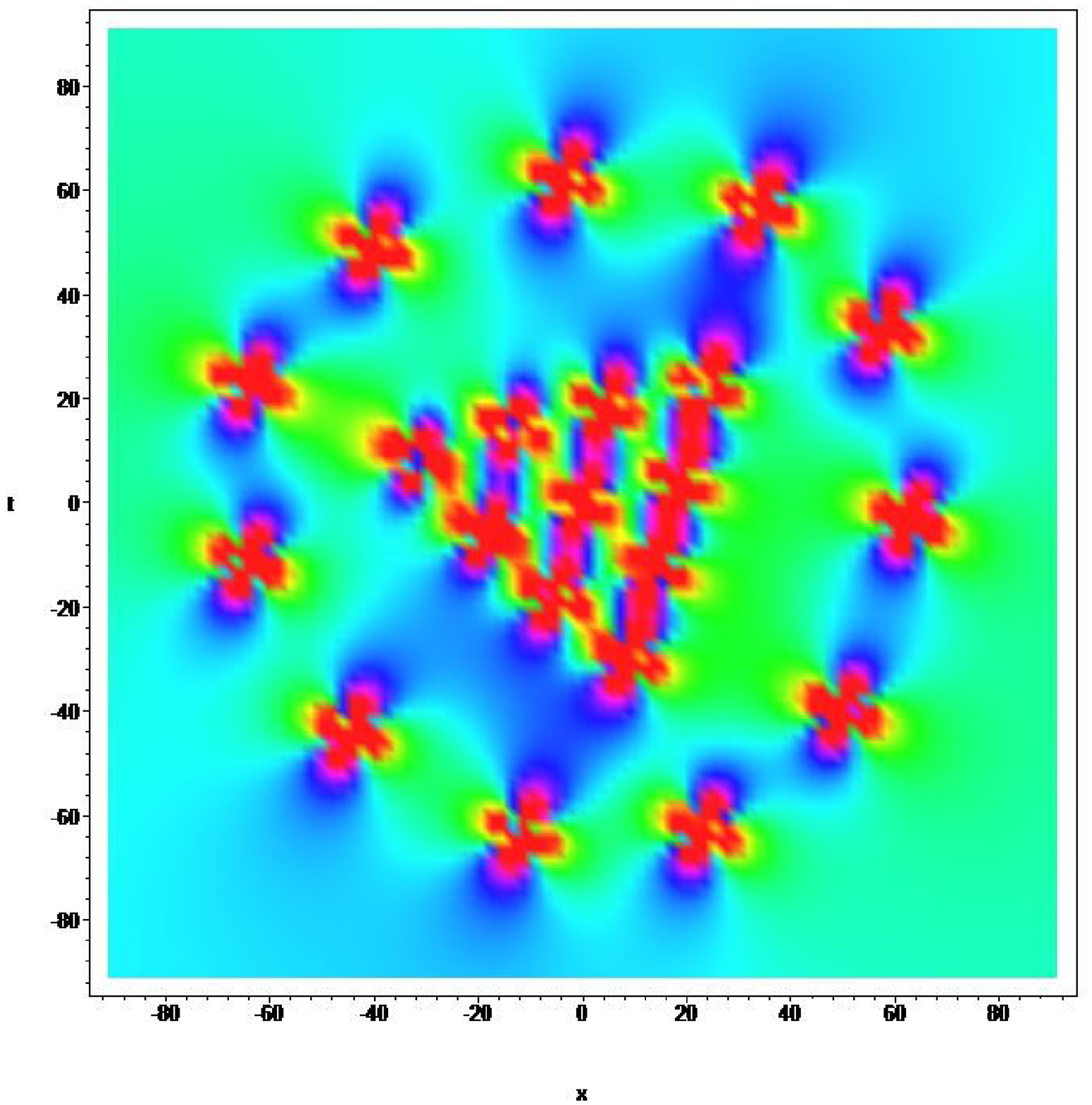}}
\qquad
\subfigure[]{\includegraphics[height=4cm,width=4cm]{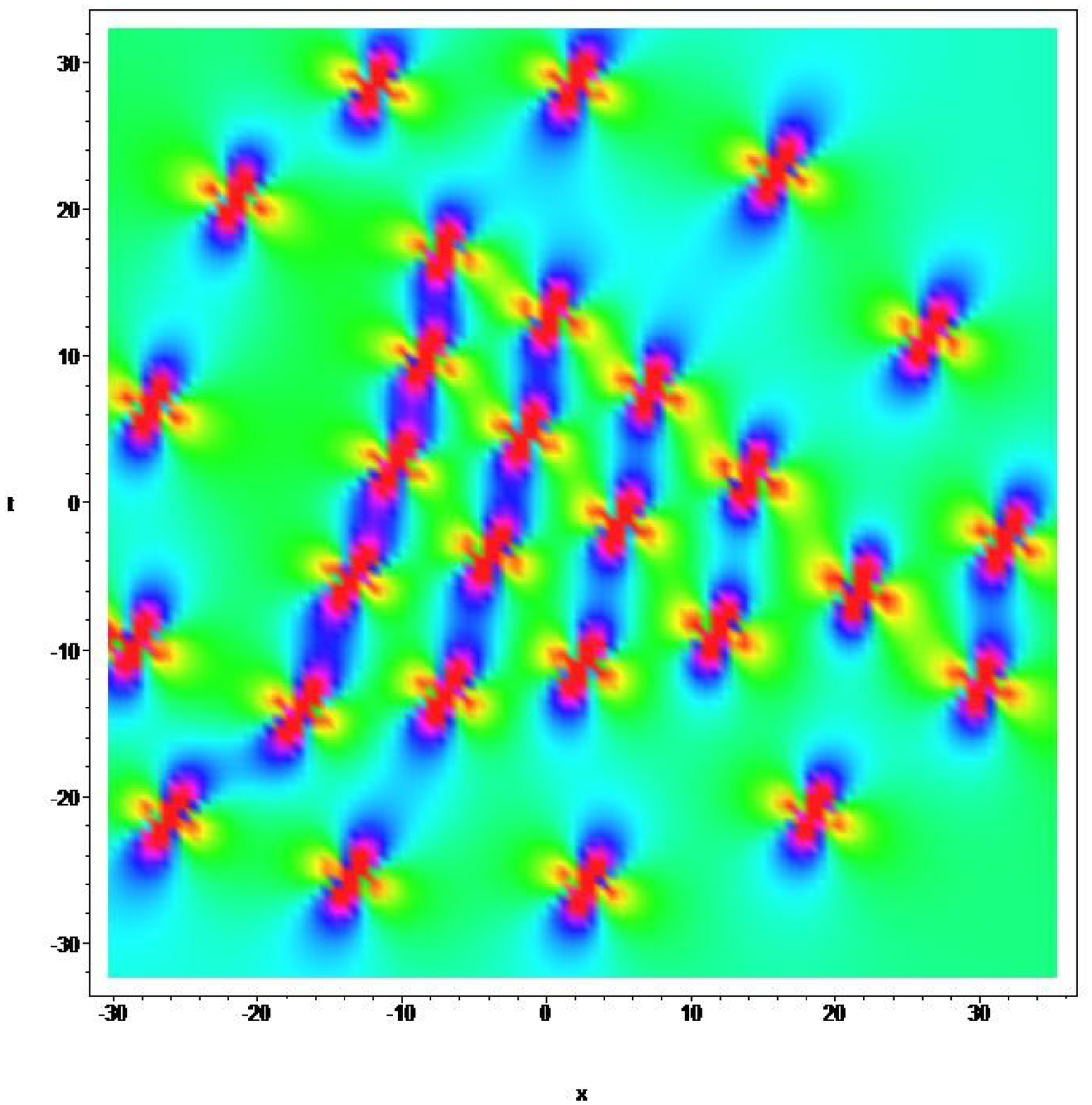}}
\caption{(Color online)The ring-triangle structure of higher order rogue wave. (a) The forth order rogue wave with $S_3=5\times10^7$, $S_1=500$. (b) The fifth order rogue wave with $S_4=5\times10^9$, $S_1=500$. (c)The sixth order rogue wave with $S_6=5\times10^{13}$, $S_1=1000$. (d) The seventh order rogue wave with $S_7=1\times10^{11}$, $S_1=150$.}\label{fig.ringtriangle}
\end{figure*}

\begin{figure*}[!htbp]
\centering
\subfigure{\includegraphics[height=4cm,width=4cm]{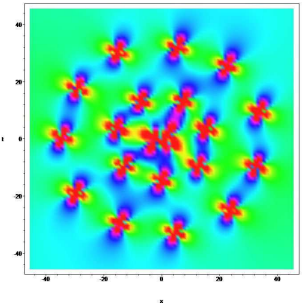}}
\caption{(Color online)The multi-ring model of the sixth order rogue wave with $S_5=8\times10^{9}$, $S_3=80000$, the inner peak is still a higher order wave}\label{fig.multiring1}
\end{figure*}

\begin{figure*}[!htbp]
\centering
\subfigure[The entire structure]{\includegraphics[height=4cm,width=4cm]{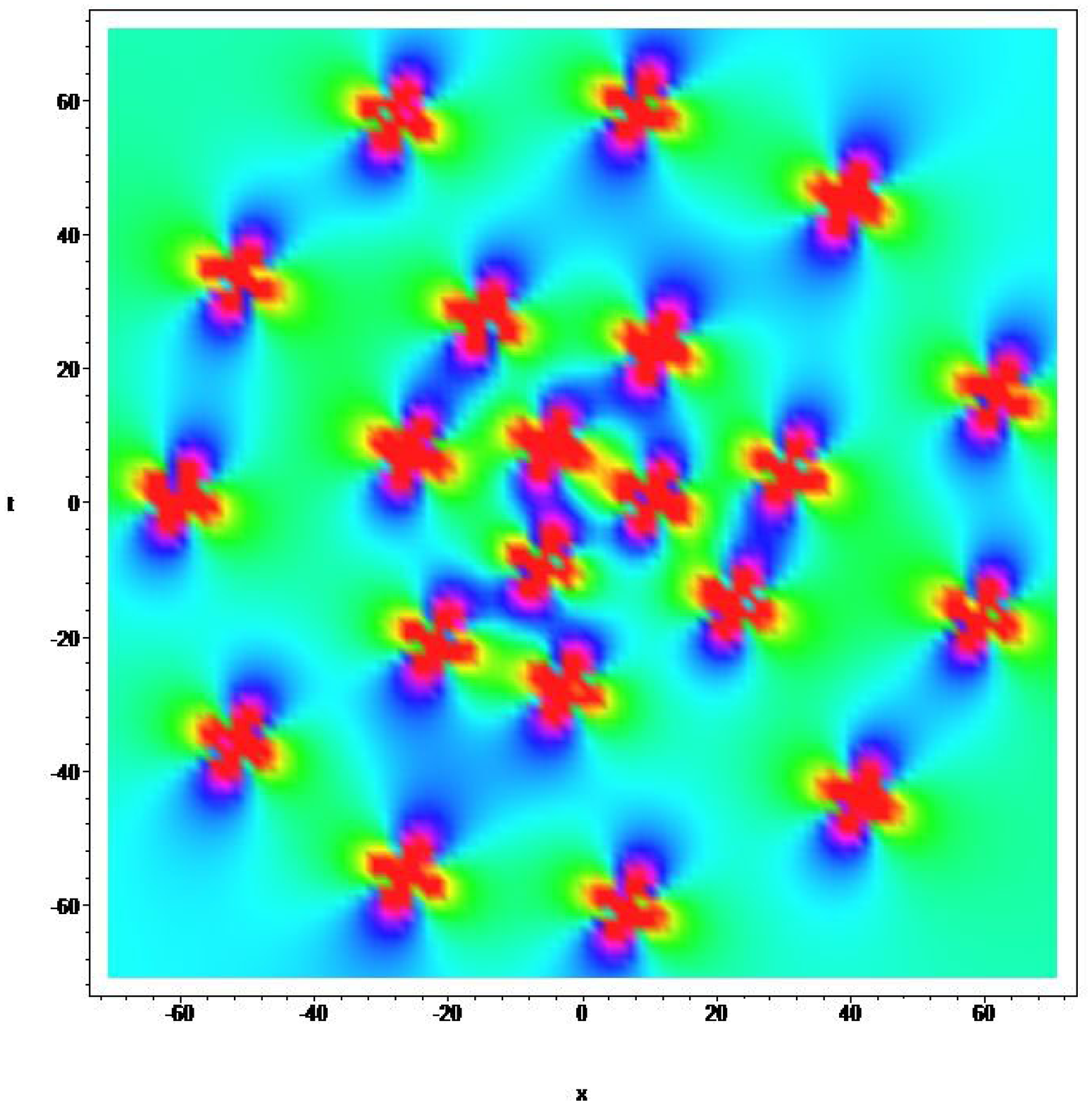}}
\qquad
\subfigure[The inner structure]{\includegraphics[height=4cm,width=4cm]{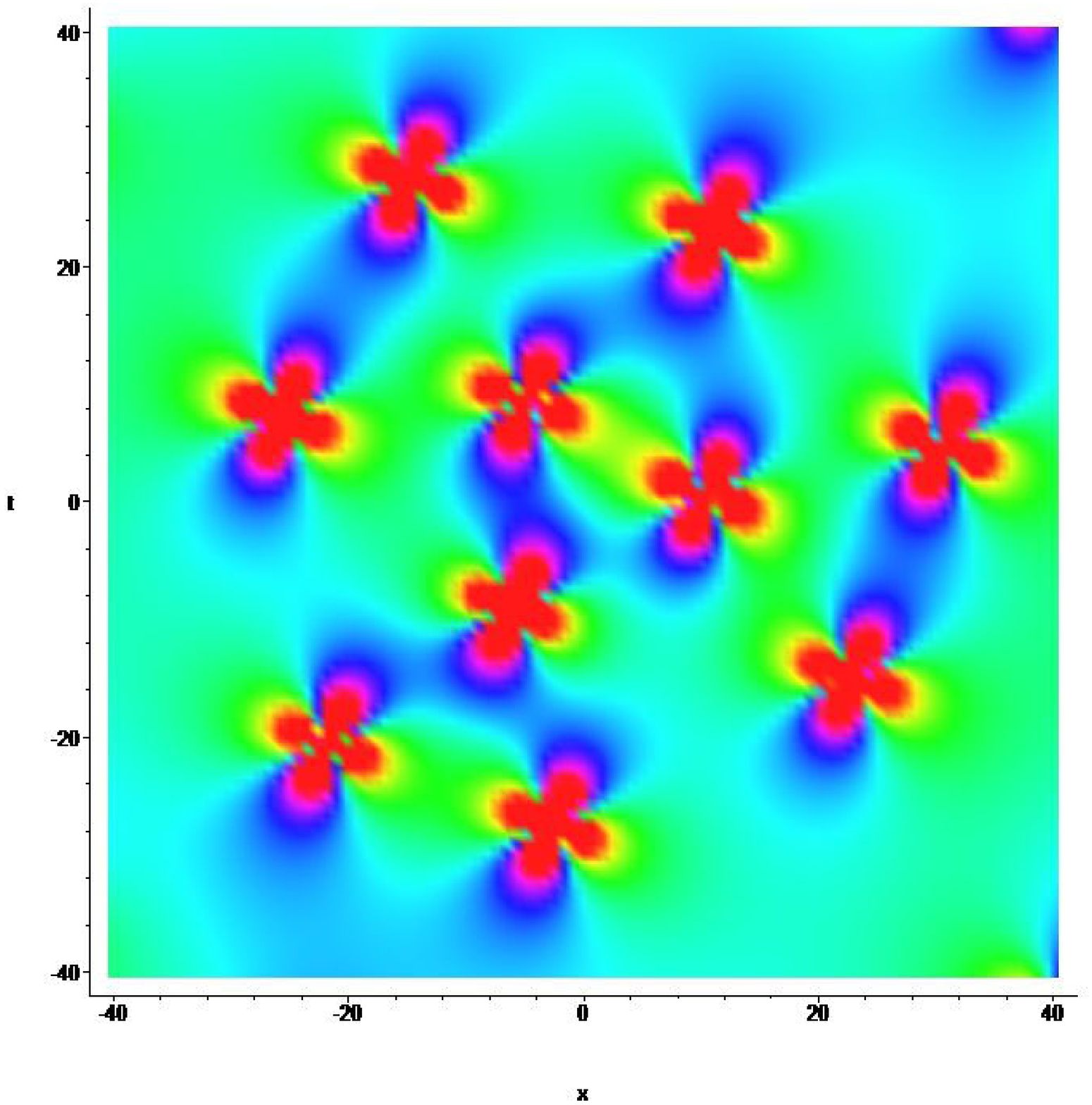}}
\caption{(Color online)The multi-ring model of the sixth order rogue wave, the inner higher rogue wave is split into a triangle.}\label{fig.multiring2}
\end{figure*}

\begin{figure*}[!htbp]
\centering
\subfigure[The entire structure]{\includegraphics[height=4cm,width=4cm]{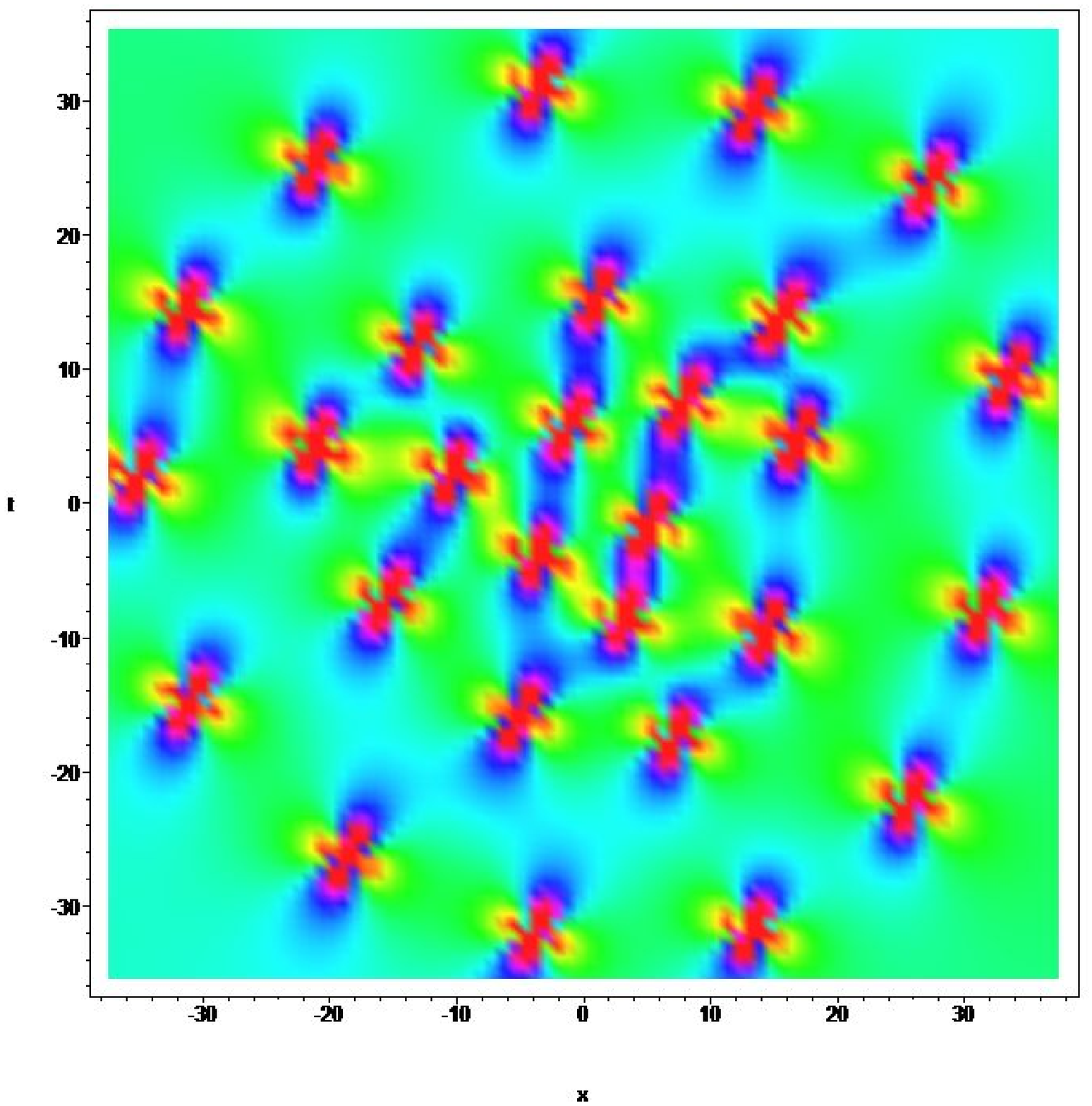}}
\qquad
\subfigure[The inner structure]{\includegraphics[height=4cm,width=4cm]{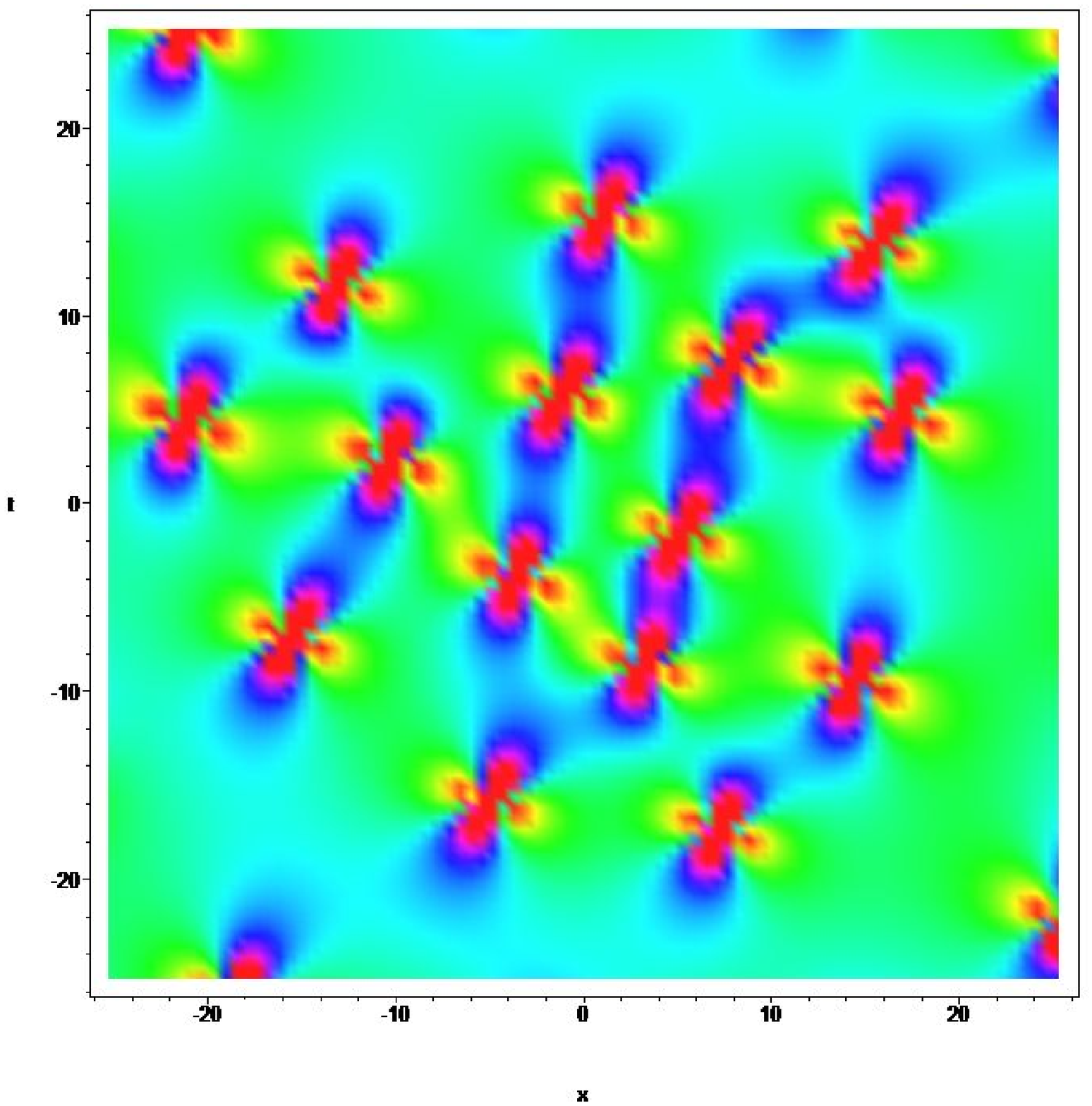}}
\\
\subfigure[The entire structure]{\includegraphics[height=4cm,width=4cm]{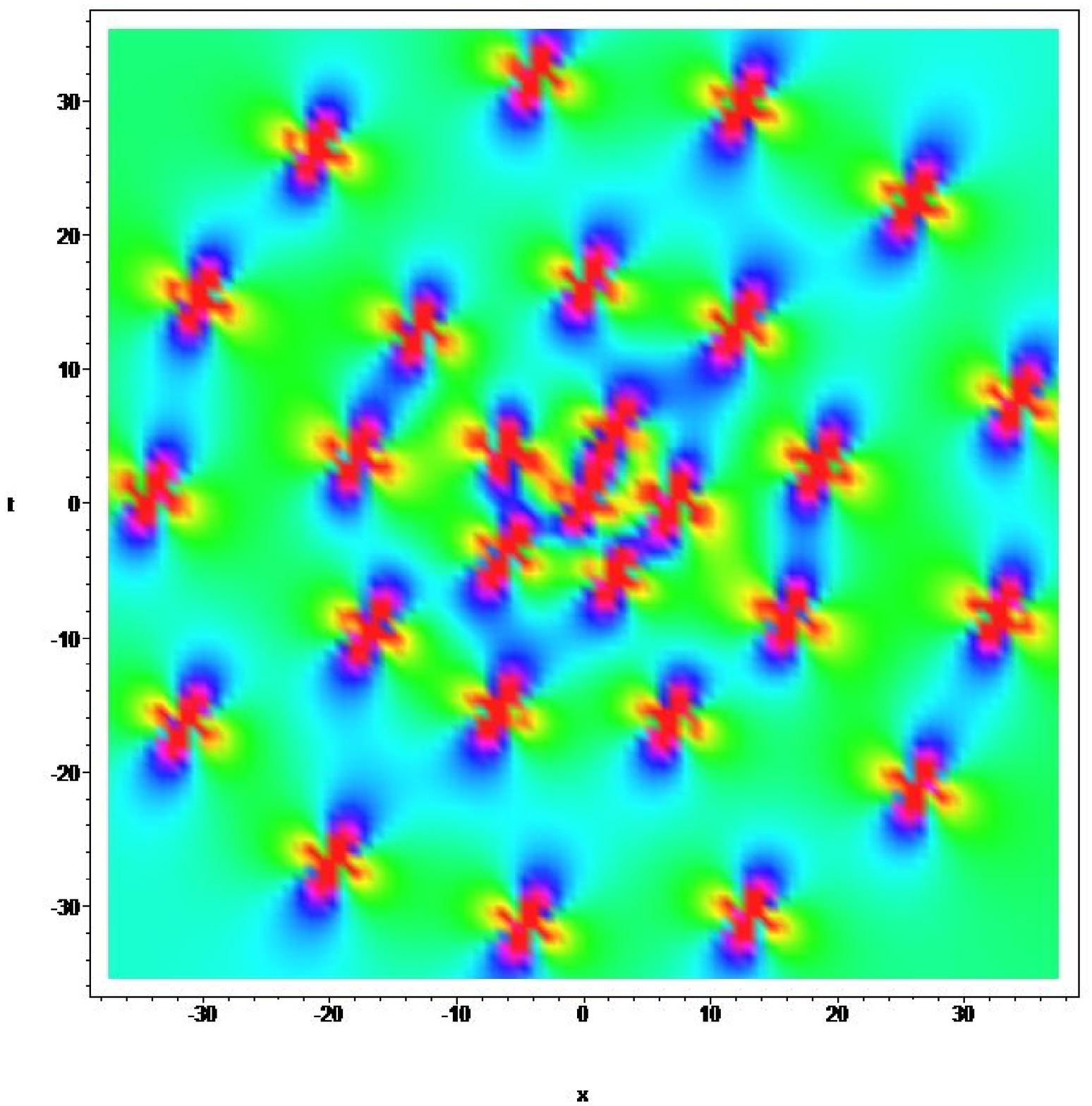}}
\qquad
\subfigure[The inner structure]{\includegraphics[height=4cm,width=4cm]{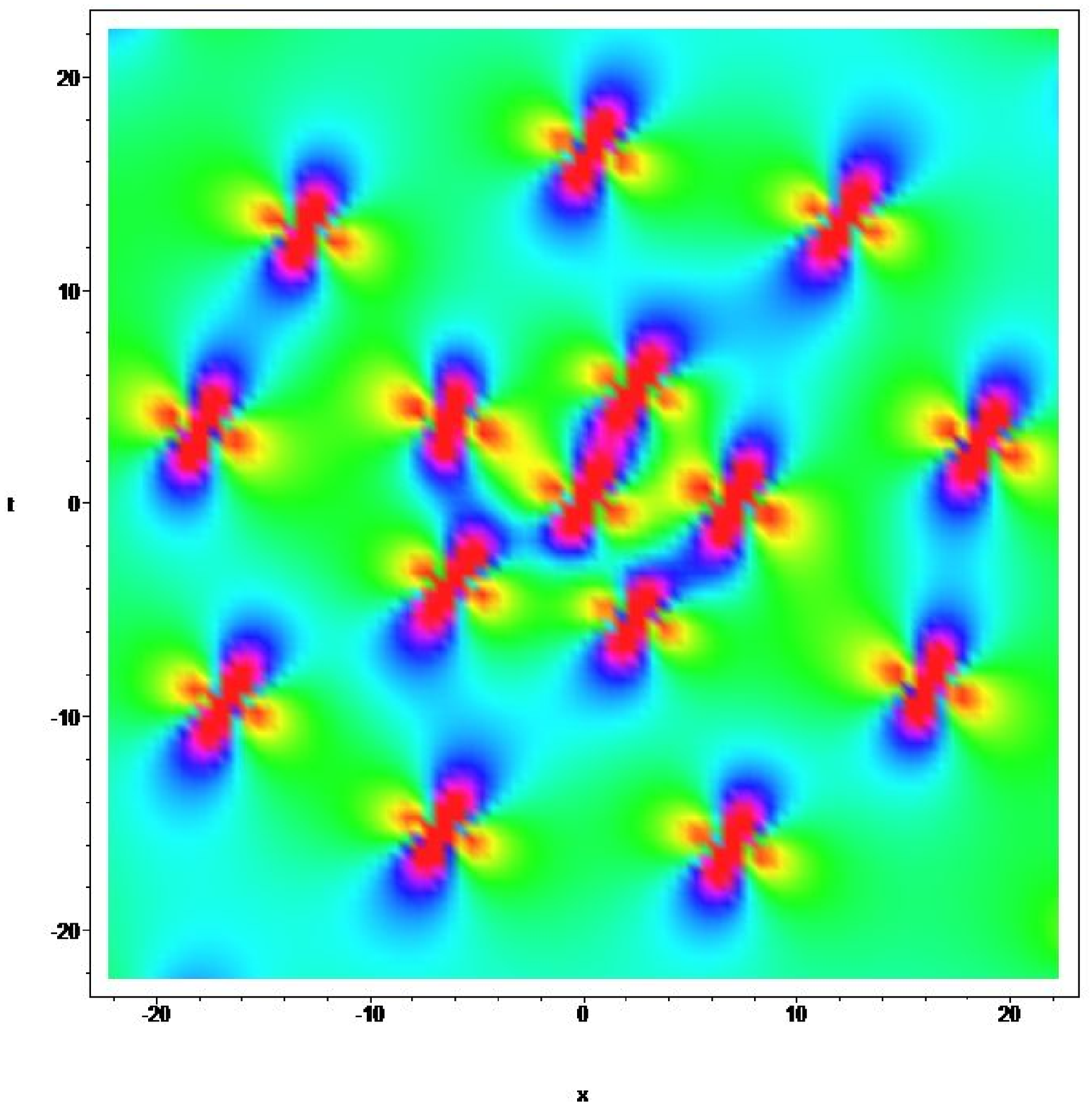}}
\caption{(Color online)The two kinds of multi-ring structure of $7$-order rogue wave solution with three parameters. (a) The $7$-order rogue wave with $S_6=5\times10^{11}$, $S_4=1\times10^7$, $S_1=100$. (b)The $7$-th order rogue wave with $S_6=5\times10^{11}$, $S_4=1\times10^7$, $S_2=1000$.}\label{fig.multiring3}
\end{figure*}

\end{document}